%% file: main.tex
\documentclass[11pt]{article}

\usepackage{amsmath}
\usepackage{amsthm}
\usepackage{amssymb}
\usepackage{algorithm}
\usepackage{color}
\usepackage[english]{babel}

\usepackage{graphicx}
\usepackage{caption}
\usepackage{subcaption}

\usepackage{enumitem}
\usepackage{version}
\usepackage[square,sort,comma,numbers]{natbib}
\usepackage{wrapfig,epsfig}
\usepackage{psfrag}
\usepackage{epstopdf}
\usepackage{url}
\usepackage{graphicx}
\usepackage{color}
\usepackage{epstopdf}
\usepackage{algpseudocode}
\usepackage[hidelinks,pdfencoding=auto,psdextra]{hyperref}
\usepackage{hyperref}
\hypersetup{colorlinks=true}
\usepackage[margin=1in]{geometry}
\newtheorem{theorem}{Theorem}[section]
\newtheorem{lemma}[theorem]{Lemma}
\newtheorem{question}[theorem]{Question}
\newtheorem{definition}[theorem]{Definition}

\newtheorem{corollary}[theorem]{Corollary}

\newtheorem{remark}[theorem]{Remark}
\newtheorem{claim}[theorem]{Claim}

\newcommand{\bi}{\boldsymbol{i}}

\newcommand{\wh}{\widehat}
\newcommand{\wt}{\widetilde}

\newcommand{\eps}{\varepsilon}
\newcommand{\Ot}{\widetilde{O}}

\newcommand{\abs}[1]{|#1|}

\newcommand{\supp}{\mathsf{supp}}
\newcommand{\poly}{\mathsf{poly}}

\DeclareMathOperator*{\E}{\mathbb{E}}

\newcommand{\FF}{\mathcal{F}}
\renewcommand{\i}{\mathbf{i}}

\DeclareMathOperator{\rect}{rect}
\DeclareMathOperator{\sinc}{sinc}

\newenvironment{proofof}[1]{\bigskip \noindent {\it Proof of #1.}\quad }
{\qed\par\vskip 4mm\par}

\input{libtheorems.tex}

\includeversion{version:full}
\excludeversion{version:short}
	
\begin{document}

\title{Estimating the Frequency of a Clustered Signal}
\author{   Xue Chen\thanks{Supported by research funding from Northwestern University. Part of this work was done while the author was in the University of Texas at Austin supported by NSF Grant CCF-1526952 and a Simons Investigator Award (\#409864, David Zuckerman).}\\
  \texttt{xue.chen1@northwestern.edu}\\
  Northwestern University
 \and   Eric Price\thanks{Supported in part by NSF Award CCF-1751040 (CAREER).}\\
  \texttt{ecprice@cs.utexas.edu}\\
  The University of Texas at Austin
}

%\date{}

\maketitle

\begin{abstract}
  We consider the problem of locating a signal whose frequencies are ``off grid" and 
  clustered in a narrow band.  Given noisy sample access to a function
  $g(t)$ with Fourier spectrum in a narrow range
  $[f_0 - \Delta, f_0 + \Delta]$, how accurately is it possible to
  identify $f_0$?  We present generic conditions on $g$ that allow for
  efficient, accurate estimates of the frequency.  We then show bounds
  on these conditions for $k$-Fourier-sparse signals that imply
  recovery of $f_0$ to within $\Delta + \Ot(k^3)$ from samples on
  $[-1, 1]$.  This improves upon the best previous bound of
  $O\big( \Delta + \Ot(k^5) \big)^{1.5}$.  We also show
  that no algorithm can do better than $\Delta + \Ot(k^2)$.

  In the process we provide a new $\Ot(k^3)$ bound on the ratio between the
  maximum and average value of continuous $k$-Fourier-sparse signals,
  which has independent application.
\end{abstract}

%\setcounter{page}{0}
%\pagebreak

\input{intro}

\input{preli}

% Section 2 about notation and Section 3 ---proof overview

\input{filter}
% Filter functions

\input{alg}
% frequency estimation algorithm

\input{sparse}

% k-sparse

\section*{Acknowledgement}
We thank Daniel Kane and Zhao Song for many helpful discussions. We also thank the anonymous referee for the detailed feedback and comments.

\bibliographystyle{alpha}
\bibliography{polynomial_interpolation}

\input{append}

\input{append2}

\end{document}

%% file: libtheorems.tex
\usepackage{etoolbox}

\makeatletter

\newcommand{\define}[4][ignore]{%
  \ifstrequal{#1}{ignore}{}{
  \@namedef{thmtitle@#2}{#1}}%
  \@namedef{thm@#2}{#4}%
  \@namedef{thmtypen@#2}{lemma}%
  \newtheorem{thmtype@#2}[theorem]{#3}%
  \newtheorem*{thmtypealt@#2}{#3~\ref{#2}}%
}

\newcommand{\state}[1]{%
  \@namedef{curthm}{#1}
  \@ifundefined{thmtitle@#1}{
  \begin{thmtype@#1}
    }{
  \begin{thmtype@#1}[\@nameuse{thmtitle@#1}]
  }
    \label{#1}
    \@nameuse{thm@#1}
  \end{thmtype@#1}
  \@ifundefined{thmdone@#1}{
  \@namedef{thmdone@#1}{stated}%
  }{}
}

\newcommand{\restate}[1]{%
  \@namedef{curthm}{#1}
  \@ifundefined{thmtitle@#1}{
    \begin{thmtypealt@#1}
    }{
  \begin{thmtypealt@#1}[\@nameuse{thmtitle@#1}]
  }
    \@nameuse{thm@#1}
  \end{thmtypealt@#1}
  \@ifundefined{thmdone@#1}{
  \@namedef{thmdone@#1}{stated}%
  }{}
}

\newcommand{\thmlabel}[1]{
  \@ifundefined{thmdone@\@nameuse{curthm}}{\label{#1}
    }{\tag*{\eqref{#1}}}
}
\makeatother

%% file: intro.tex
\section{Introduction}

% First: k-sparse fourier superresolution, Prony & Moitra and PS and such
% Second: can appx. _signal_ when arbitrarily close
% But what about frequency of that signal?  CKPS gave poly(k)/T approx; 
% And then: what about other, non-sparse signals w/ clustered frequencies?
% E.g., sine * exp(-C t).  Frequency band +/- C has most of the energy.  

A natural question, dating at least to the work of Prony in 1795, is
to estimate a signal from samples, assuming the signal has a
$k$-sparse Fourier representation, i.e., that the signal is a sum of
$k$ complex exponentials:
$g(t) = \sum_{j=1}^k v_j e^{2\pi \mathbf{i} f_j t}$ for some set of
frequencies $f_j$ and coefficients $v_j$.

If the frequencies are located on a discrete grid (giving a sparse
discrete Fourier transform), then a long line of work has studied
efficient algorithms for recovering the
signal (e.g., \cite{Man92,GGIMS,AGS,GMS,HIKP12,IK}).  If the frequencies are
not on a grid, then Prony's method from 1795~\cite{Prony} or matrix
pencil~\cite{BM86} can still identify them in the absence of noise.
With noise, however, one cannot robustly recover frequencies that are
too close together: if one listens to a signal for the interval
$[-T, T]$ then any two frequencies $\theta$ and $\theta + \eps/T$ will
be $O(\eps)$-close to each other, and so cannot be distinguished with
noise.  As shown in~\cite{Moitra}, this nonrobustness grows
exponentially in $k$.  On the other hand,~\cite{Moitra} also showed
that recovery with polynomially small noise is possible if all the
frequencies have separation $1/2T$, and~\cite{PS15} showed that a
constant fraction of noise is tolerable with separation
$\log^{O(1)}(FT)/T$, where $F$ is the bandlimit of the signal.

So what \emph{is} possible for arbitrary Fourier-sparse signals,
without any assumption of frequency separation?  One cannot hope to
identify the frequencies exactly, but one can still estimate the
\emph{signal itself}.  If two frequencies are similar enough to be
indistinguishable over the sampled interval, we do not need to
distinguish them.  In~\cite{CKPS17}, this led to an algorithm for
an arbitrary $k$-Fourier-sparse signal that used $\text{poly}(k, \log(FT))$ samples to estimate it with only a constant factor increase in the
noise.  However, this polynomial is fairly poor.

Since prior work could handle the case of well-separated frequencies,
a key challenge in~\cite{CKPS17} is the setting with all the
frequencies in a narrow cluster.  Formally, consider the following
subproblem: if all the frequencies $f_i$ of the signal lie in a narrow
band $[f_0 - \Delta, f_0 + \Delta]$, how accurately can we estimate
$f_0$?  Note that while we would like an efficient algorithm that
takes a small number of samples, the key question is \emph{information
  theoretic}.  And we can ask this question more generally: if the
signal is not $k$-sparse, but still has all its frequencies in a narrow
band, can we locate that band?

\begin{question}
  Let $g(t)$ be a signal with Fourier transform supported on
  $[f_0 - \Delta, f_0 + \Delta]$, for some $f_0 \in [-F, F]$.  Suppose that
  we can sample from $y(t) = g(t) + \eta(t)$ at points in $[-T, T]$,
  where $\eta(t)$ could be any $\ell_2$ bounded noise on $[-T,T]$ with 
  \[
    \E_{t \in [-T, T]} \big[ |\eta(t)|^2\big] \leq \eps \E_{t \in [-T, T]} \big[ |g(t)|^2\big]
  \]
  for a small constant $\eps$.  Under what conditions on $g$ can we
  estimate $f_0$, and how accurately?
\end{question}

One might expect to be able to estimate $f_0$ to
$\pm (\Delta + O(\frac{1}{T}))$ for all functions $g$; after all, $g$ is
just a combination of individual frequencies, each of which points to
some frequency in the right range, and each individual frequency in
isolation can be estimated to within $\pm O(\frac{1}{T})$ in the
presence of noise.  Unfortunately, this intuition is false.

To see this, consider the family of $k$-sparse Fourier functions with
$f_j = \eps j$, i.e.,
\[
  \text{span}(e^{2\pi\mathbf{i} (j \eps)t} \mid j \in [k]).
\]
By sending $\eps \to 0$ and taking a Taylor expansion, this family can
get arbitrarily close to any degree $k-1$ polynomial, on any interval
$[-T', T']$.  Thus, to solve the question, one would also need to
solve it when $g(t)$ is a polynomial even for arbitrarily small
$\Delta$.

There are two ways in which $g(t)$ being a degree $d$ polynomial can
lead to trouble.  The first is that $g(t)$ could itself be a Taylor
expansion of $e^{\pi \mathbf{i} f t}$.  If $d \gtrsim f T$, this Taylor
approximation will be quite accurate on $[-T, T]$; with the noise
$\eta$, the observed signal can \emph{equal} $e^{\pi \mathbf{i} f t}$.
Thus the algorithm has to output $f$, which can be $\Theta(d/T)$ far
from the ``true'' answer $f_0 = 0$.

The second way in which $g(t)$ can lead to trouble is by removing most
of the signal energy.  If $g(t)$ is the (slightly shifted) Chebyshev
polynomial $g(t) = T_d\big( t/T + O(\frac{\log^2 d}{d^2}) \big)$, then
$\abs{g(t)} \leq 1$ for $t \leq \big(1 - O(\frac{\log^2 d}{d^2}) \big) T$, while
$g(t) \geq d$ for $t \geq \big(1 - O(\frac{\log^2 d}{d^2})\big) T$.  That is to
say, the majority of the $\ell_2$ energy of $g$ can lie in the final
$O(\frac{\log^2 d}{d^2})$ fraction of the interval.  In such a case, a
small constant noise level $\eta$ can make samples outside that
$T\cdot \Ot(1/d^2)$ size region equal to zero, and hence completely
uninformative; and samples in that region still have to tolerate
noise.  This leads to an ``effective'' interval size of
$T' = T \cdot \Ot(\frac{1}{d^2})$, leading to accuracy
$O(1/T') = \Ot(d^2)/T$.

Our main result is that, in a sense, these two types of difficulties
are the only ones that arise.  We can measure the second type of
difficulty by looking at how much larger the maximum value of $g$ is
than its average:
\[
  R := \frac{\sup_{t \in [-T, T]} \abs{g(t)}^2}{\E_{t \in [-T, T]} \abs{g(t)}^2}.
\]
We can measure the former by observing that while a polynomial may
approximate a complex exponential on a bounded region, as
$t \to \infty$ the polynomial will blow up.  In particular, we take
the $S$ such that
\[
  |g(t)|^2 \leq \poly(R) \cdot \E_{t \in [-T,T]} \big[ |g(t)|^2 \big] \cdot |\frac{t}{T}|^{S}
\]
for all $\abs{t} \geq T$.  We show that if $R$ and $S$ are bounded,
one can estimate $f_0$ to within $\Delta + \Ot(R + S)/T$, which is almost tight from the above discussion of polynomials.  Moreover,
the time and number of samples required are fairly efficient:

\define{thm:general_bound}{Theorem}{
Given any $T>0, F>0, \Delta>0,R$, and $S>0$, let $g(t)$ be a signal with the following properties:
\begin{enumerate}
\item $\supp(\wh{g}) \subseteq [f_0-\Delta,f_0 + \Delta]$ where $f_0 \in [-F,F]$.
\item $\underset{t \in [-T,T]}{\sup} \big[ |g(t)|^2 \big] \le R \cdot \underset{t \in [-T,T]}{\E} \big[ |g(t)|^2 \big] $.
\item $|g(t)|^2$ grows as at most $\poly(R) \cdot \underset{t \in [-T,T]}{\E} \big[ |g(t)|^2 \big] \cdot |\frac{t}{T}|^{S}$ for $t \notin [-T,T]$. 
\end{enumerate}   
Let $y(t)=g(t)+\eta(t)$ be the observable signal on $[-T,T]$, where $\underset{t \in [-T,T]}{\E} \big[ |\eta(t)|^2 \big] \le \epsilon \cdot \underset{t \in [-T,T]}{\E} \big[ |g(t)|^2 \big]$ for a sufficiently small constant $\epsilon$. For $\Delta'=\Delta + \frac{\wt{O}(R + S)}{T}$ and any $\delta>0$, there exists an efficient algorithm that takes $O(R \log \frac{F}{\Delta' \cdot \delta})$ samples from $y(t)$ and outputs $\wt{f}$ satisfying $|f_0-\wt{f}| \le O(\Delta')$ with probability at least $1-\delta$.
}

\state{thm:general_bound}

\paragraph{Application to sparse Fourier transforms}
Specializing to $k$-Fourier-sparse signals, we give bounds on $R$ and
$S$ for this family.  Since (as described above) this family can
approximate degree-$(k-1)$ polynomials, we know that $R \gtrsim k^2$
and $S \gtrsim k$; we show that $R \lesssim k^3 \log^2 k$ and
$S \lesssim k^2 \log k$.  Thus, whenever $R$ is between $k^2$ and
$\Ot(k^3)$, we can identify $k$-Fourier-sparse signals to within
$\Delta + \Ot(R)/T$.  This is an improvement over the results
in~\cite{CKPS17} in several ways.

Formally, for a given sparsity level $k$, we consider signals in
$$
\mathcal{F} :=\left\{g(t)=\sum_{j=1}^k v_j e^{2 \pi \bi f_j t} \bigg| f_j \in [-F,F] \right\}.$$

\define{thm:worst_case_sFFT}{Theorem}{ 
For any $k$ and $T$, 
\begin{equation}\label{eq:def_R}
R := \sup_{g \in \mathcal{F}} \frac{\underset{x \in [-T,T]}{\sup} |g(x)|^2}{\underset{x \in [-T,T]}{\E}[|g(x)|^2]} = O(k^3 \log^2 k).
\end{equation}
}

\state{thm:worst_case_sFFT}

It was previously known that
$R \lesssim k^4 \log^{3} k$~\cite{CKPS17}, and this fact was used
in~\cite{AKMMVZ18}.  (Thus, our improved bound on $R$ immediately
implies an improvement in Theorem 8 of~\cite{AKMMVZ18}, from
$s_{\mu,\eps}^5 \log^3 s_{\mu,\eps}$ to
$s_{\mu,\eps}^4 \log^2 s_{\mu,\eps}$.)

Next we bound the growth $S=\wt{O}(k^2)$ for any $|t| \ge T$.
\define{thm:bound_growth_around1}{Theorem}{
There exists $S=O(k^2 \log k)$ such that for any $|t|>T$ and $g(t)=\sum_{j=1}^k v_j \cdot e^{2 \pi \bi f_j t}$, $|g(t)|^2 \le \poly(k) \cdot \underset{x \in [-T,T]}{\E}[|g(x)|^2] \cdot |\frac{t}{T}|^{S}$.
}

\state{thm:bound_growth_around1}

This is analogous to Theorem 5.5 of~\cite{CKPS17}, which proves a
bound of $(kt)^k$ rather than $t^{\wt{O}(k^2)}$.  These bounds are
incomparable, but the $t^{\wt{O}(k^2)}$ bound is actually more useful for this
problem: what really matters is showing that $g(t)$ is not too large
just outside the interval.  Theorem~\ref{thm:bound_growth_around1}
gives the ``correct'' polynomial dependence at $t = (1 + 1/k^2)T$.

We can now apply Theorem~\ref{thm:general_bound} to get an efficient
algorithm to recover the center of a cluster of $k$ frequencies within
accuracy $\tilde{O}(R)$.

\define{thm:sparse_freq_estimation}{Theorem}{
  Given $F,T,$ and $k$, let $R$ be the ratio between the maximum and average value of continuous $k$-Fourier-sparse signals defined in \eqref{eq:def_R}. Given $\Delta$, let $g(t)$ be a $k$-Fourier-sparse signal centered around $f_0$: $g(t)=\sum_{i
    \in [k]} v_i \cdot e^{2 \pi \i f_i t}$ where each $f_i \in
  [f_0-\Delta,f_0 + \Delta]$ and
  $y(t)=g(t)+\eta(t)$ be the observable signal on $[-T,T]$, where $\underset{t \in [-T,T]}{\E} \big[ |\eta(t)|^2 \big] \le \epsilon \cdot \underset{t \in [-T,T]}{\E} \big[ |g(t)|^2 \big]$ for a sufficiently small constant $\epsilon$.

For any $\delta>0$, there exist $\Delta'=\Delta + \frac{\tilde{O}(R)}{T}$ and an
efficient algorithm that takes $O(k \log^2 k \log \frac{F}{\Delta' \cdot
  \delta})$ samples from $y(t)$ and outputs $\wt{f}$ satisfying
$|f_0-\wt{f}| \le O(\Delta')$ with probability at least $1-\delta$.
}

\state{thm:sparse_freq_estimation}

Note that the sample complexity here is $\Ot(k)$ not $\Ot(R)$.  This
is because, based on the structure of the problem, we can use a
nonuniform sampling procedure that performs better.  Otherwise this
theorem is just Theorem~\ref{thm:general_bound} applied to the $R$ and
$S$ from Theorems~\ref{thm:worst_case_sFFT}
and~\ref{thm:bound_growth_around1}.

Theorem~\ref{thm:sparse_freq_estimation} is a direct improvement on
Theorem 7.5 of~\cite{CKPS17}, which for $T=1$ could estimate to within
$O\left(\Delta + \wt{O}(k^5)\right)^{1.5}$ accuracy and used
$\text{poly}(k)$ samples.  In particular, in addition to improving the
additive $\text{poly}(k)$ term, our result avoids a multiplicative
increase in the bandwidth $\Delta$ of $g$.

The main technical lemma in proving Theorems~\ref{thm:general_bound} and~\ref{thm:sparse_freq_estimation} is
a filter function $H$ with a compact supported Fourier transform $\wh{H}$ that simulates a
box function on $[-T,T]$ for any $g$ satisfying the conditions in
Theorem~\ref{thm:general_bound}.

\define{lem:effect_H_k_sparse}{Lemma}{
Given any $T$, $S$, and $R$, there exists a filter function $H$ with $\big|\supp(\wh{H})\big| \le \frac{\tilde{O}(R+S)}{T}$ such that for any $g(t)$ satisfying the second and third conditions in Theorem~\ref{thm:general_bound},
\begin{enumerate}
\item $H$ is close to a box function on $[-T,T]$: $\int_{-T}^T |g(t) \cdot H(t)|^2 \mathrm{d} t \ge 0.9 \int_{-T}^T |g(t)|^2 \mathrm{d} t$.
\item The tail of $H(t) \cdot g(t)$ is small: $\int_{-T}^T |g(t) \cdot H(t)|^2 \mathrm{d} t \ge 0.95 \int_{-\infty}^{\infty} |g(t) \cdot H(t)|^2 \mathrm{d} t.$
\end{enumerate}
}

\state{lem:effect_H_k_sparse}

\paragraph{Organization} 
We introduce some notation and tools in Section~\ref{sec:preli}. Then we provide a technical overview in Section~\ref{sec:overview}. We show our filter function and prove Lemma~\ref{lem:effect_H_k_sparse} in Section~\ref{sec:filter_func}. Next we present the algorithm about frequency estimation of Theorem~\ref{thm:general_bound} in Section~\ref{sec:freq_estimation}. Finally we prove the results about sparse Fourier transform --- Theorem~\ref{thm:worst_case_sFFT} and Theorem~\ref{thm:bound_growth_around1} in Section~\ref{sec:condition_num_growth}.

%% file: preli.tex
\section{Preliminaries}\label{sec:preli}
In the rest of this work, we fix the observation interval to be $[-1,1]$ and define 
\begin{equation}\label{eq:def_l2}
\|g\|_2=\big( \underset{x \sim [-1,1]}{\E} |g(x)|^2 \big)^{1/2},
\end{equation} because we could rescale $[-T,T]$ to $[-1,1]$ and $[-F,F]$ to $[-FT,FT]$. 

We first review several facts about the Fourier transform. The Fourier transform $\wh{g}(f)$ of an integrable function $g:\mathbb{R} \rightarrow \mathbb{C}$ is
$$
\wh{g}(f)=\int_{-\infty}^{+\infty} g(t)e^{-2\pi\bi f t} \mathrm{d} t \text{ for any real } f.
$$
We use $g \cdot h$ to denote the pointwise dot product $g(t) \cdot h(t)$ and $g^{k}$ to denote $\underbrace{g(t) \cdots g(t)}_{k}$. Similarly, we use $g * h$ to denote the convolution of $g$ and $h$: $\int_{-\infty}^{+\infty} g(x) \cdot h(t-x) \mathrm{d} x$. In this work, we always set $g^{*k}$ as the convolution $\underbrace{g(t) * \cdots * g(t)}_{k}$. Notice that $\supp(g \cdot h)=\supp(g) \cap \supp(h)$ and $\supp(g * h)=\supp(g)+\supp(h)$.

We define the box function and its Fourier transform $\sinc$ function as follows. Given a width $s>0$, the box function $\rect_s(t)=1/s$ iff $|t| \le s/2$; and its Fourier transform is $\sinc(sf)=\frac{\sin(\pi fs)}{\pi f s}$ for any $f$. 

We state the Chernoff bound for random sampling \cite{chernoff1952}.
\begin{lemma}\label{lem:chernoff_bound}
Let $X_1, X_2, \cdots, X_n$ be independent random variables in $[0,R]$ with expectation $1$. For any $\eps<1/2$ and $n \gtrsim \frac{R}{\epsilon^2}$, $X=\frac{\sum_{i=1}^n X_i}{n}$ with expectation 1 satisfies 
\[
\Pr[|X-1| \ge \eps] \le 2 \exp(-\frac{\eps^2}{3} \cdot \frac{n}{R}).
\]
\end{lemma}

\section{Proof Overview}\label{sec:overview}
We first outline the proofs of Lemma~\ref{lem:effect_H_k_sparse} and Theorem~\ref{thm:general_bound}. Then we show the proof sketch of $R=\tilde{O}(k^3)$ and $S=\tilde{O}(k^2)$ of $k$-Fourier-sparse signals.

\paragraph{The filter functions $(H,\wh{H})$ in Lemma~\ref{lem:effect_H_k_sparse}.} Ideally, to satisfy the two claims in Lemma~\ref{lem:effect_H_k_sparse}, we could set $H(t)$ to be the box function $2 \rect_2(t)$ on $[-1,1]$. However, by the uncertainty principle, it is impossible to make its Fourier transform $\wh{H}$ compact using such an $H(t)$. Hence our construction of $(H,\wh{H})$ is in the inverse direction: we build $\wh{H}(f)$ by box functions and $H(t)$ by the Fourier transform of box functions --- the sinc function. In the rest of this discussion, we focus on using the sinc function to prove Lemma~\ref{lem:effect_H_k_sparse} given the properties of $g$ in Theorem~\ref{thm:general_bound}.

We first notice that any $H$ with the following two properties is effective in Lemma~\ref{lem:effect_H_k_sparse} for $g$ satisfying $|g(t)|^2 \le R \cdot \|g\|_2^2$ for any $|t| \le 1$ and $|g(t)|^2 \le \poly(R) \|g\|^2_2 \cdot |t|^S$ for $|t|>1$:
\begin{enumerate}
\item $H(t) = 1 \pm 0.01$ for any $t \in [-1+\frac{1}{C \cdot R},1-\frac{1}{C \cdot R}]$ of a large constant $C$. This shows
$$\int_{-1}^1 |H(t) \cdot g(t)|^2 \mathrm{d} t \ge 0.99^2 \int_{-1+\frac{1}{C \cdot R}}^{1-\frac{1}{C \cdot R}} |g(t)|^2 \mathrm{d} t.
$$
Because $|g(t)|^2 \le R \cdot \|g\|_2^2$ for any $t \in [-1,1] \setminus [-1+\frac{1}{C \cdot R},1-\frac{1}{C \cdot R}]$, the constant on the R.H.S. is at least $0.99^2 \cdot (1-\frac{1}{C}) \ge 0.9$, which implies the first claim of Lemma~\ref{lem:effect_H_k_sparse}.
\item $H(t)$ declines to $\frac{1}{\poly(R) \cdot t^{2S}}$ for any $|t|>1$. This shows
$$
\int_{1}^{\infty} |H(t) \cdot g(t)|^2 \mathrm{d} t \le 0.01 \int_{-1}^1 |g(t)|^2 \mathrm{d} t,
$$ which implies the second claim.
\end{enumerate}
 
For ease of exposition, we start with $S=0$. We plan to design a filter $H_0(t)$ with compact $\wh{H_0}$ dropping from $0.99$ at $t=1-\frac{1}{C \cdot R}$ to $\frac{1}{\poly(R)}$ at $t=1$ in a small range $\frac{1}{C R}$ using the sinc function. To apply the sinc function, we notice that 
$$
\sinc(C R \cdot t)^{O(\log R)}=\left( \frac{\sin(\pi C R \cdot t)}{\pi C R \cdot t} \right)^{O(\log R)}
$$ decays from 1 at $t=0$ to $1/\poly(R)$ at $t = \frac{1}{C \cdot R}$, which matches the dropping of $H_0(t)$ from $t=1-\frac{1}{C \cdot R}$  to $t=1$. 

Then, to make $H(t) \approx 1$ for any $|t|\le 1-\frac{1}{C \cdot R}$, let us consider a convolution of $\rect_{1}(t)$ and $\sinc(C R \cdot t)^{O(\log R)}$. Because most of the mass of the latter is in $[-\frac{1}{C R},\frac{1}{C R}]$, this convolution keeps almost the same value in $[-\frac{1}{2}+\frac{1}{C R},\frac{1}{2}-\frac{1}{C R}]$ and drops down to $1/\poly(R)$ at $t=\frac{1}{2}+\frac{1}{C R}$. At the same time, it will keep the compactness of $\wh{H_0}$ since it corresponds to the dot product on the Fourier domain.  By normalizing and scaling, this gives the desired $(H_0,\wh{H_0})$ for $S=0$.

Next we describe the construction of $S>0$. The high level idea is to consider the decays of $H(t)$ in $\log_2 S + O(1)$ segments rather than one segment of $S=0$: 
$$
(1-\frac{1}{C R}, 1], (1, 1+\frac{1}{S}], (1+\frac{1}{S}, 1+\frac{2}{S}], \ldots, (1+\frac{2^j}{S}, 1+\frac{2^{j+1}}{S}], \ldots, (1+\frac{S/2}{S}, 2], (2, +\infty).
$$
For each segment, we provide a power of sinc functions matching its decay in $H(t)$ like the construction of $H_0$ on $(1-\frac{1}{C R}, 1]$. The final construction is the convolution of the dot product of all sinc powers and a box function, which appears in Section~\ref{sec:filter_func}.

\paragraph{The Algorithm of Theorem~\ref{thm:general_bound}.} Now we show
how to estimate $f_0$ given the observable signal $y=g+\eta$ where
$\supp(\wh{g}) \subseteq [f_0 - \Delta, f_0+\Delta]$ and
$\|\eta\|_2^2 \le \eps \|g\|_2^2$ (with $\ell_2$ norm taken over
$[-T, T]$ defined in \eqref{eq:def_l2}). We instead consider $y_H(t)=y(t) \cdot H(t)$ with the
filter function $(H,\wh{H})$ from Lemma~\ref{lem:effect_H_k_sparse}
and the corresponding dot products $g_H=g \cdot H$ and
$\eta_H=\eta \cdot H$. The starting point is that for a sufficiently
small $\beta$, we expect
\[
  y_H(t+\beta) \approx e^{2 \pi \bi f_0 \beta} \cdot y_H(t)
\]
because $y_H$ has Fourier spectrum concentrated around $f_0$.  This
does not hold for \emph{all} $t$, but it does hold on average:
\begin{align}
\int_{-1}^1 |y_H(t+\beta) - e^{2 \pi \bi f_0 \beta} \cdot y_H(t)|^2
\mathrm{d} t \lesssim \eps \cdot \int_{-1}^1 |y_H(t)|^2 \mathrm{d} t.\label{eq:shift}
\end{align}
This is because we can use Parseval's identity to replace these
integrals by an integral over Fourier domain---Parseval's identity
would apply if the integrals were from $-\infty$ to $\infty$, but
because of the filter function $H$, relatively little mass in $y_H$
lies outside $[-1, 1]$.  Then, the Fourier transform of the term
inside the left square is
$e^{2 \pi \bi f \beta} \cdot \wh{y_H}(f) - e^{2 \pi \bi f_0 \beta}
\cdot \wh{y_H}(f)$.  Note that $\wh{y_H}=\wh{g_H}+\wh{\eta_H}$ has
most of its $\ell_2$ mass in
$\supp(g_H) \subseteq [f_0-\Delta',f_0+\Delta']$ for
$\Delta'=\Delta+|\supp(\wh{H})|$, and every such frequency shrinks in
the left by a factor
$|e^{2 \pi \bi f \beta} - e^{2 \pi \bi f_0 \beta}| = O(\beta \Delta')$.  Thus, for
$\beta \ll 1/\Delta'$,~\eqref{eq:shift} holds.

To learn $f_0$ through $e^{2 \pi \bi f_0 \beta}$, we design a sampling procedure to output $\alpha$ satisfying 
$$
|y_H(\alpha+\beta) - e^{2 \pi \bi f_0 \beta} y_H(\alpha)| \le 0.3 \cdot y_H(\alpha) \text{ with probability \emph{more than half}}.
$$ 
Even though the above discussion shows the left hand side is smaller
than the R.H.S. on average, a uniformly random $\alpha \sim [-1,1]$
may not satisfy it with good probability:
$|y_H(\alpha)| \ge \|y_H\|_2$ may be only true for $1/R$ fraction of
$\alpha \in [-1,1]$, while the corruption by adversarial noise $\eta$ has $\|\eta\|_2^2 \gtrsim \eps \|y_H\|_2^2$ for a constant
$\eps \gg 1/R$. At the same time, even for many points
$\alpha_1,\ldots,\alpha_m$ where some of them satisfy the above
inequality, it is infeasible to verify such an $\alpha_i$ given $f_0$
is unknown. We provide a solution by adopting the importance sampling:
for $m=O(R)$ random samples $\alpha_1,\ldots,\alpha_m \in [-1,1]$, we
output $\alpha$ with probability proportional to the weight
$|y_H(\alpha_i)|^2$.

We prove the correctness of this sampling procedure in Lemma~\ref{lem:good_sampling} in Section~\ref{sec:freq_estimation}.

Finally, learning $e^{2\pi \mathbf{i} f_0 \beta}$ is not enough to
learn $f_0$: because of the noise, we only learn
$e^{2\pi \mathbf{i} f_0 \beta}$ to within a constant $\eps$, which
gives $f_0$ to within $\pm O(\eps/\beta)$; and because of the
different branches of the complex logarithm, this is only up to
integer multiples of $1/\beta$.  Therefore to fully learn $f_0$, we
repeat the sampling procedure at logarithmically many different scales
of $\beta$, from $\beta=1/2F$ to $\beta = \frac{\Theta(1)}{\Delta'}$.

\paragraph{$k$-Fourier-sparse signals.} Finally, we show $R=\wt{O}(k^3)$ and $S=\wt{O}(k^2)$ such that for any $g(t)=\sum_{j=1}^k v_j \cdot e^{2 \pi \bi f_j t}$ --- not necessarily one with the $f_j$ clustered together---
$$
\frac{\underset{t \in [-1,1]}{\sup} |g(t)|^2}{\|g\|_2^2} \le R \text{ and } |g(t)|^2 \le \poly(R) \cdot \|g\|_2^2 \cdot |t|^S.
$$

We first review the previous argument of $R=\wt{O}(k^4)$ \cite{CKPS17}. The key point is to show for some $d=\wt{O}(k^2)$ that $g(1)$ is a linear combination of $g(1-\theta),\ldots,g(1- d \cdot \theta)$ using bounded integer coefficients $c_1,\ldots,c_d=O(1)$ for any $\theta \le \frac{2}{d}$.  Then
\begin{equation}\label{eq:comb_g}
g(1)=\sum_{j \in [d]} c_j \cdot g(1-j \cdot \theta) \text{ implies } |g(1)|^2 \le (\sum_{j \in [d]} |c_j|^2) \cdot (\sum_{j \in [d]} |g(1-j \cdot \theta)|^2 ).
\end{equation}

If we think of $g(1)$ as the supremum and 
$g(1-j \cdot \theta)$ as the average $\|g\|_2$---which we can formally
do up to logarithmic factors by averaging over $\theta$---this shows
$|g(1)|^2 \le \wt{O}(d^2) \|g\|_2^2$. One natural idea to improve it
is to use a smaller value $d$ and a shorter linear combination
\cite{CP18}. However, $d=\tilde{\Omega}(k^2)$ for such a combination
when $g$ is approximately the degree $k-1$ Chebyshev polynomial. In
this work, we use a geometric sequence to control $c_j$ such that
$\sum_j |c_j|^2 = O(d/k)$ instead of $O(d)$, which provides an
improvement of a factor $\wt{O}(k)$ on $R$.

Then we bound $S=\wt{O}(k^2)$ for $g(t)$ at $|t|>1$. The intuition is that given \eqref{eq:comb_g} holds for any $g(t)$ in terms of $g(t-\theta),\ldots,g(t-d \cdot \theta)$ with $\theta=\frac{2}{d}$, it implies $|g(t)|^2 \le \poly(k) \cdot \|g\|_2^2 \cdot e^{(t-1) \cdot O(d)}$ for $t > 1$. Combining this with an alternate bound $|g(t)|^2 \le \poly(k) \cdot \|g\|_2^2 \cdot (k \cdot t)^{O(k)}$ for $t > 1 + 1/k$, it completes the proof of Theorem~\ref{thm:bound_growth_around1} about $S$.

Finally we notice that we could improve the sample complexity in Theorem~\ref{thm:sparse_freq_estimation} to $\wt{O}(k) \log \frac{F}{\Delta'}$ using a biased distribution~\cite{CP18} to generate $\alpha$. These results about $k$-Fourier-sparse signals appear in Section~\ref{sec:condition_num_growth}.

%% file: filter.tex
\section{Our Filter Function}\label{sec:filter_func}
The main result is an explicit filter function $H$ with compact support $\wh{H}$ that  is close to the box function on $[-1,1]$ for any $g$ satisfying the conditions in Theorem~\ref{thm:general_bound}.

We show our filter function as follows. 
\begin{definition}\label{def:filter_func_general}
Given $R$, the growth rate $S$ and an even constant $C$, we define the filter function
$$
H(t)= s_0 \cdot \left( \sinc(C R \cdot t)^{C \log R} \cdot \sinc\big(C \cdot S \cdot t\big)^{C} \cdot \sinc\big(\frac{C \cdot S}{2} \cdot t\big)^{2C} \cdots \sinc\big(C \cdot t \big)^{C \cdot S} \right)*\rect_2(t)
$$ where $s_0 \in \mathbb{R}^+$ is a parameter to normalize $H(0)=1$. On the other hand, its Fourier transform is
$$
\wh{H}(f) = s_0 \cdot \left( \rect_{C R}(f)^{*C \log R} * \rect_{C \cdot S}(f)^{*C} * \rect_{\frac{C \cdot S}{2}}(f)^{*2C} * \cdots * \rect_{C}(f)^{*C S} \right) \cdot \sinc(2t),
$$
whose support size is $O(CR \cdot C \log R + C S \cdot C + \cdots + C \cdot C \cdot S)=O(R \log R + S \log S)$.
\end{definition}

%We review a few facts about the Fourier transform over $\mathbb{R}$. Given a width $s>0$, the box function $\rect_s(t)=1/s$ iff $|t| \le s/2$; and its Fourier transform is $\sinc(sf)=\frac{\sin(\pi fs)}{\pi f s}$ for any $f$. 
%\begin{definition}
%Given the sparsity $k$, the upper bound $R=\underset{k \text{-Fourier-sparse } g}{\sup} \frac{\underset{x \in [-1,1]}{\sup} |g(x)|^2}{\|g\|_2^2}$ and a constant $C$, we define a filter function 
%$$
%H(t)= s_0 \cdot \left( \sinc(C R \cdot t)^{C \log k} \cdot \sinc\big(C \cdot k^2 \log k \cdot t\big)^{C} \cdot \sinc\big(\frac{C \cdot k^2 \log k}{2} \cdot t\big)^{2C} \cdots \sinc\big(k \cdot t\big)^{C^2 \cdot k \log k} \right)*\rect_2(t)
%$$ with a fixed normalizer $s_0 \in \mathbb{R}^+$ such that $H(0)=1$. On the other hand, its Fourier transform is
%$$
%\wh{H}(f) = s_0 \cdot \left( \rect_{C R}(f)^{*C \log k} * \rect_{C \cdot k^2 \log k}(f)^{*C} * \rect_{\frac{C \cdot k^2 \log k}{2}}(f)^{*2C} * \cdots * \rect_{k}(f)^{*C^2 \cdot k \log k} \right) \cdot \sinc(2f),$$
%whose support size is $O(CR \cdot C \log k + C k^2 \log k \cdot C + \cdots + k \cdot C^2 \cdot k \log k)=\tilde{O}(R)$.
%\end{definition}

We prove Lemma~\ref{lem:effect_H_k_sparse} using $H(\alpha x)$ with a large constant $C$ and a scale parameter $\alpha=\frac{1}{2}+\frac{1.2}{\pi C R}$. For convenience, we state the full version of Lemma~\ref{lem:effect_H_k_sparse} for $T=1$ as follows.

\begin{theorem}\label{thm:HwithScaling}
  Let $R, S > 0$, let $C$ be a large even constant, and define
  $\alpha=(\frac{1}{2}+\frac{1.2}{\pi C R})$. Consider any function
  $g$ satisfying the following two conditions:
\begin{enumerate}
\item $\underset{t \in [-1,1]}{\sup} |g(t)|^2 \le R \cdot \|g\|_2^2 $
\item And $|g(t)|^2 \le \poly(R) \cdot \|g\|^2_2 \cdot |t|^{S}$ for $t \notin [-1,1]$, 
\end{enumerate}
Then the filter function $H\big(\alpha x \big)$ is such that
$H\big( \alpha x \big) \cdot g(x)$ satisfies
\begin{enumerate}
\item $\int_{-1}^1 |g(x) \cdot H\big( \alpha x \big)|^2 \mathrm{d} x \ge 0.9 \int_{-1}^1 |g(x)|^2 \mathrm{d} x$.
\item $\int_{-1}^1 |g(x) \cdot H\big( \alpha x \big)|^2 \mathrm{d} x \ge 0.95 \int_{-\infty}^{\infty} |g(x) \cdot H\big( \alpha x \big)|^2 \mathrm{d} x.$
\item $|H(x)| \le 1.01$ for any $x$.
\end{enumerate}
\end{theorem}

\begin{version:full}
For completeness, we show a few properties of $H$ and finish the proof of Theorem~\ref{thm:HwithScaling} in Appendix~\ref{appen:filter}.
\end{version:full}

\begin{version:short}
Due to the space constraint, we defer the proof of Theorem~\ref{thm:HwithScaling} to the full version.
\end{version:short}

%% file: alg.tex
\section{Frequency Estimation}\label{sec:freq_estimation}
We show the algorithm for frequency estimation and prove Theorem~\ref{thm:general_bound} in this section. We fix $T=1$ and use the definition $\|h\|_2^2 = \underset{x \sim [-1,1]}{\E}[|h(x)|^2]$ to restate the theorem.

\begin{theorem}\label{thm:rescale_general}
Given any $F>0, \Delta>0,R$, and $S>0$, let $g(t)$ be a signal with the following properties:
\begin{enumerate}
\item $\supp(\wh{g}) \subseteq [f_0-\Delta,f_0 + \Delta]$ where $f_0 \in [-F,F]$.
\item $\underset{t \in [-1,1]}{\sup} \big[ |g(t)|^2 \big] \le R \cdot \|g\|_2^2$. %\underset{t \in [-1,1]}{\E} \big[ |g(t)|^2 \big] 
\item $|g(t)|^2$ grows as at most $\poly(R) \cdot \|g\|^2_2 \cdot |t|^{S}$ for $t \notin [-1,1]$. 
\end{enumerate}   
Let $y(t)=g(t)+\eta(t)$ be the observable signal on $[-1,1]$, where $\|\eta\|_2^2 \le \epsilon \cdot \|g\|_2^2$ for a sufficiently small constant $\epsilon$. For $\Delta'=\Delta + \wt{O}(R + S)$ and any $\delta$, there exists an efficient algorithm that takes $O(R \log \frac{F}{\Delta' \cdot \delta})$ samples from $y(t)$ and outputs $\wt{f}$ satisfying $|f_0-\wt{f}| \le O(\Delta')$ with probability at least $1-\delta$.
\end{theorem}

For convenience, we set $h_H(t)=h(t) \cdot H(\alpha t)$ for any signal $h(t)$ with the filter function $H$ defined in Theorem~\ref{thm:HwithScaling} such that $y_H(t)=y(t) \cdot H(\alpha t)$. 

Given the observation $y(t)$ with most Fourier mass concentrated around $f_0$, the main technical result in this section is an estimation of $e^{2 \pi \i \beta f_0}$ through $y_H(\alpha) e^{2 \pi i f_0 \beta} \approx y_H(\alpha+\beta)$.

\begin{lemma}\label{lem:good_sampling}
Given parameters $F,R,S$,  and $\Delta$, let $g$ be a signal satisfying the three conditions in Theorem~\ref{thm:general_bound} for some $f_0 \in [-F,F]$ and $\Delta'=\Delta + O(R \log R + S \log S)$.
    
Let $y(t)=g(t)+\eta(t)$ be the observable signal on $[-1,1]$ where the noise $\|\eta\|^2_2 \le \epsilon \|g\|^2_2$ for a sufficiently small constant $\epsilon$. There exist a constant $\gamma$ and an algorithm such that for any $\beta \le \frac{\gamma}{\Delta'}$, it takes $O(R)$ samples to output $\alpha$ satisfying $|y_H(\alpha) e^{2 \pi i f_0 \beta}-y_H(\alpha+\beta)| \le 0.3 |y_H(\alpha)|$ with probability at least 0.6.
\end{lemma}

We show our algorithm in Algorithm~\ref{alg:freq_est}. We finish the proof of Theorem~\ref{thm:sparse_freq_estimation} here and defer the proof of Lemma~\ref{lem:good_sampling} to Section~\ref{sec:proof_good_sampling}.

\begin{algorithm}
\caption{Obtain one good $\alpha$}\label{alg:freq_est}
\begin{algorithmic}[1]
\Procedure{\textsc{ObtainOneGoodSample}}{$R,y(t)$}
\State Let $m=C \cdot R$ for a large constant $C$.
\State Take $m$ random samples $x_1,\cdots,x_m$ uniform in $[-1,1]$.
\State Query $y(x_i)$ and compute $y_H(x_i)=y(x_i) \cdot H(x_i)$ for each $i$.
\State Set a distribution $D_m$ proportional to $|y_H(x_i)|^2$, i.e., $D_m(x_i)=\frac{|y_H(x_i)|^2}{\sum_{j=1}^m |y_H(x_j)|^2}$.
\State Output $\alpha \sim D_m$.
\EndProcedure
\end{algorithmic}
\end{algorithm}

\begin{proofof}{Theorem~\ref{thm:rescale_general}}
From Lemma~\ref{lem:good_sampling}, $\frac{y_H(\alpha+\beta)}{y_H(\alpha)}$ gives a good estimation of $e^{2 \pi i f_0 \beta}$ with probability 0.6 for any $\beta \le \frac{\gamma}{\Delta'}$. We use the frequency search algorithm of Lemma 7.3 in \cite{CKPS17} with the sampling procedure in Lemma~\ref{lem:good_sampling}. Because the algorithm in \cite{CKPS17} uses the sampling procedure $O(\log \frac{F}{\Delta' \cdot \delta})$ times to return a frequency $\wt{f}$ satisfying $|\wt{f} - f_0| \le \Delta'$ with prob. at least $1-\delta$, the sample complexity is $O(R \cdot \log \frac{F}{\Delta' \cdot \delta})$.
\end{proofof}

\subsection{Proof of Lemma~\ref{lem:good_sampling}}\label{sec:proof_good_sampling}
For $y_H(x)=g_H(x) + \eta_H(x)$, we have the following concentration lemma for estimation $g_H(x)$.

\begin{claim}\label{clm:generate_sampling_after_filter}
Given any $g$ satisfying the three conditions in Theorem~\ref{thm:general_bound} and any $\eps$ and $\delta$, there exists $m=O(R \log \frac{1}{\delta}/\eps^2)$ such that for $m$ random samples $x_1,\ldots,x_m \sim [-1,1]$, with probability $1-\delta$,
$$
\frac{\sum_{i=1}^m |g_H(x_i)|^2 }{m} \in [1-\eps,1+\eps] \cdot \E_{x \sim [-1,1]}[|g_H(x)|^2].
$$
\end{claim}
\begin{proof}
Notice that $\frac{\underset{x \sim [-1,1]}{\sup} [|g_H(x)|^2]}{\underset{x \sim [-1,1]}{\E}[|g_H(x)|^2]} \le 2 R$. From the Chernoff bound in Lemma~\ref{lem:chernoff_bound}, $m=O(R \log \frac{1}{\delta}/\eps^2)$ suffices to estimate $\|g_H\|_2^2$.
\end{proof}

Next we consider the effect of noise $\eta_H(x_i)$ and $y_H(x_i)$. 
\begin{claim}\label{clm:samples_y_H}
With probability $0.9$ over $m$ random samples in $[-1,1]$, $\sum_{i=1}^m |y_H(x_i)|^2/m \ge 0.8 \|g\|_2^2$.
\end{claim}
\begin{proof}
From Theorem~\ref{thm:HwithScaling}, $\|g_H\|_2^2 \ge 0.95 \|g\|_2^2$. Thus Claim~\ref{clm:generate_sampling_after_filter} implies $\sum_{i=1}^m |g_H(x_i)|^2/m \ge 0.98 \cdot 0.95 \|g\|_2^2$ for $m=O(R)$ with probability 0.99. 

At the same time, because $\E[\sum_{i=1}^m |\eta_H(x_i)|^2/m]=\|\eta_H\|_2^2$, $\sum_{i=1}^m |\eta_H(x_i)|^2/m \le 14 \|\eta_H\|_2^2$ with probability at least $1-\frac{1}{14}$ from the Markov inequality. This is also less than $14 \cdot 1.02^2 \|\eta\|_2^2 \le 15 \epsilon \|g\|_2^2$ from the upper bound on $H(t)$. 

We have 
$$
\frac{1}{m} \sum_{i=1}^m |y_H(x_i)|^2 \ge \frac{1}{m} \sum_{i=1}^m \bigg(|g_H(x_i)|^2 - 2 |g_H(x_i)| \cdot |\eta_H(x_i)| + |\eta_H(x_i)|^2\bigg).
$$ By the Cauchy-Schwartz inequality, the cross term $\sum_{i=1}^m |g_H(x_i)| \cdot |\eta_H(x_i)| \le (\sum_{i=1}^m |g_H(x_i)|^2)^{1/2} \cdot (\sum_{i=1}^m |\eta_H(x_i)|^2)^{1/2}$. From all discussion above, 
$$
\frac{1}{m} \sum_{i=1}^m |y_H(x_i)|^2 \ge \big(0.93 - 2 \sqrt{0.93 \cdot 15 \epsilon}\big) \|g\|_2^2.
$$
When $\eps$ is a small constant, it is at least $0.8 \cdot \|g\|_2^2$.
\end{proof}

We set $z(t)=y_H(t) \cdot e^{2 \pi \bi f_0 \beta} - y_H(t+\beta)$ for convenience and bound it as follows.

\begin{claim}
Given any small constant $\gamma$, $\Delta' = \Delta + \supp(H)$, and $z(t)=y_H(t) \cdot e^{2 \pi \bi f_0 \beta} - y_H(t+\beta)$ for $\beta \le \frac{\gamma}{\Delta'}$,  $\|z\|_2^2 \lesssim (\gamma^2 + \epsilon) \|g\|_2^2$.
\end{claim}
\begin{proof}
Notice that $y_H=g_H+\eta_H$ where $\supp(\wh{g_H}) \in [f_0 - \Delta, f_0 + \Delta]$ such that $$\int_{f \notin [f_0 - \Delta',f_0 + \Delta']} |\wh{y}(f)|^2 \mathrm{d} f \le \int_{-\infty}^{\infty} |\wh{\eta_H}(f)|^2 \mathrm{d} f = \int_{-\infty}^{\infty} |\eta_H(t)|^2 \mathrm{d} t \le 1.02^2 \epsilon \int_{-1}^{1} |g(t)|^2 \mathrm{d} t.$$

We bound $\|z\|_2^2$ through 
$$
\int_{-1}^1 |z(t)|^2 \mathrm{d} t \le \int_{-\infty}^{\infty} |z(t)|^2 \mathrm{d} t = \int_{-\infty}^{\infty} |\wh{z}(f)|^2 \mathrm{d} f = \int_{f_0 - \Delta'}^{f_0 + \Delta'} |\wh{z}(f)|^2 \mathrm{d} f + \int_{f \notin [f_0 - \Delta',f_0 + \Delta']} |\wh{z}(f)|^2 \mathrm{d} f.
$$

Therefore we write
$$
\int_{f_0 - \Delta'}^{f_0 + \Delta'} |\wh{z}(f)|^2 \mathrm{d} f = \int_{f_0 - \Delta'}^{f_0 + \Delta'} |\wh{y_H}(f) \cdot e^{2 \pi \bi f_0 \beta} - \wh{y_H}(f) \cdot e^{2 \pi \bi f \beta}|^2 \mathrm{d} f \le \int_{f_0 - \Delta'}^{f_0 + \Delta'} |\wh{y_H}(f) |^2 \cdot |e^{2 \pi \bi f_0 \beta} - e^{2 \pi \bi f \beta}|^2 \mathrm{d} f.
$$
Because $f \in [f_0 - \Delta', f_0 + \Delta']$ and $\beta \le \frac{\gamma}{\Delta'}$, $|e^{2 \pi \bi f_0 \beta} - e^{2 \pi \bi f \beta}| \le 4\pi \gamma$. So $$\int_{f_0 - \Delta'}^{f_0 + \Delta'} |\wh{z}(f)|^2 \mathrm{d} f \lesssim \gamma^2 \int_{-\infty}^{+\infty} |\wh{y_H}(f) |^2 \mathrm{d} f = \gamma^2 \int_{-\infty}^{+\infty} |y_H(t) |^2 \mathrm{d} t \lesssim \gamma^2 (1+2\epsilon) \int_{-1}^{1} |g(t)|^2 \mathrm{d} t.$$

On the other hand,
\begin{align*}
\int_{f \notin [f_0 - \Delta',f_0 + \Delta']} |\wh{z}(f)|^2 \mathrm{d} f & = \int_{f \notin [f_0 - \Delta',f_0 + \Delta']} |\wh{y_H}(f) \cdot e^{2 \pi \bi f_0 \beta} - \wh{y_H}(f) \cdot e^{2 \pi \bi f \beta}|^2 \mathrm{d} f \\
& \le 4 \int_{f \notin [f_0 - \Delta',f_0 + \Delta']} |\wh{y_H}(f)|^2 \mathrm{d} f \\
& \le 4 \int_{-\infty}^{+\infty} |\wh{\eta_H}(f)|^2  \mathrm{d} f 
 = 4 \int_{-\infty}^{+\infty} |\wh{\eta_H}(t)|^2  \mathrm{d} t
\end{align*}
which is less than $5 \epsilon \int_{-1}^{1} |g(t)|^2 \mathrm{d} t.$

From all discussion above, $\int_{-1}^1 |z(t)|^2 \mathrm{d} t \lesssim (\gamma^2 + \epsilon) \int_{-1}^{1} |g(t)|^2 \mathrm{d} t$.
\end{proof}
For sufficiently small $\gamma$ and $\eps$, by Markov inequality, we have the following corollary.
\begin{corollary}\label{cor:samples_z}
For sufficiently small constants $\gamma$ and $\epsilon$, with probability $0.9$ over $m$ random samples in $[-1,1]$, $\sum_{i=1}^m |z(x_i)|^2 \le 0.01 \|g\|_2^2$.
\end{corollary}

Finally we finish the proof of Lemma~\ref{lem:good_sampling}.

\begin{proofof}{Lemma~\ref{lem:good_sampling}}
We assume Claim~\ref{clm:samples_y_H} and Corollary~\ref{cor:samples_z} hold in this proof, i.e.,
$$
\sum_{i=1}^m |y_H(x_i)|^2/m \ge 0.8 \|g\|_2^2 \text{ and } \sum_{i=1}^m |z(x_i)|^2/m \le 0.01 \|g\|_2^2.
$$
For a random sample $\alpha \sim D_m$, we bound
\[
\E_{\alpha \sim D_m}\left[ \frac{|y_H(\alpha) e^{2 \pi i f_0 \beta} - y_H(\alpha+\beta)|^2}{|y_H(\alpha)|^2} \right]= \E_{\alpha \sim D_m}\left[ \frac{|z(\alpha)|^2}{|y_H(\alpha)|^2} \right] = \sum_{i=1}^m \frac{|z(x_i)|^2}{|y_H(x_i)|^2} \cdot \frac{|y_H(x_i)|^2}{\sum_{j=1}^m |y_H(x_j)|^2}.
\]
This is $\frac{\sum_{i=1}^m |z(x_i)|^2}{\sum_{j=1}^m |y_H(x_j)|^2} \le \frac{0.01}{0.8}$. Thus with probability $0.8$, $\frac{|y_H(\alpha) e^{2 \pi i f_0 \beta} - y_H(\alpha+\beta)|^2}{|y_H(\alpha)|^2}$ is less than $0.05/0.8 \le 0.09$. From all discussion above, $\frac{|y_H(\alpha) e^{2 \pi i f_0 \beta} - y_H(\alpha+\beta)|}{|y_H(\alpha)|} \le 0.3$ with probability 0.6.
\end{proofof}

%% file: sparse.tex
\section{Bounds on Fourier-sparse Signals}\label{sec:condition_num_growth}
We consider $g(t)=\sum_{j=1}^k v_j e^{2\pi \bi f_j t}$ where each $f_j \in [f_0-\Delta,f_0+\Delta]$ in this section. The main result is to prove $R=\tilde{O}(k^3)$ and $S=\tilde{O}(k^2)$ for $k$ arbitrary real frequencies. We restate Theorem~\ref{thm:sparse_freq_estimation} after fixing $T=1$.

\begin{theorem}\label{thm:restate_sparse}
  Given $F,\Delta,$ and $k$, let $g(t)$ be a $k$-Fourier-sparse signal centered around $f_0 \in [-F,F]$: $g(t)=\sum_{i
    \in [k]} v_i \cdot e^{2 \pi \i f_i t}$ where $f_i \in
  [f_0-\Delta,f_0 + \Delta]$ and
  $y(t)=g(t)+\eta(t)$ be the observable signal on $[-1,1]$, where $\|\eta\|_2^2 \le \epsilon \cdot \|g\|_2^2$ for a sufficiently small constant $\epsilon$.

For any $\delta>0$, there exist $\Delta'=\Delta + \tilde{O}(R)$ and an
efficient algorithm that takes $O(k \log^2 k \log \frac{F}{\Delta' \cdot
  \delta})$ samples from $y(t)$ and outputs $\wt{f}$ satisfying
$|f_0-\wt{f}| \le O(\Delta')$ with probability at least $1-\delta$.
\end{theorem}

The main improvement is a biased distribution that saves the sample complexity from $O(R) \cdot \log \frac{F}{\Delta' \cdot \delta}$ to $\wt{O}(k) \cdot \log \frac{F}{\Delta' \cdot \delta}$.

\begin{version:full}
We provide the main technical lemma here and defer the proofs of Theorem~\ref{thm:worst_case_sFFT}, \ref{thm:bound_growth_around1}, and \ref{thm:restate_sparse} to Appendix~\ref{append:sparse_signal}.
\end{version:full}

\begin{version:short}
We provide the main technical lemma here and defer the proofs of Theorem~\ref{thm:worst_case_sFFT}, \ref{thm:bound_growth_around1}, and \ref{thm:restate_sparse} to the full version.
\end{version:short}

\begin{theorem}\label{thm:sup_sup_coef}
Given $z_1,\ldots,z_k$ with 
$|z_1|=|z_2|=\cdots=|z_k|=1$, there exists a degree $d=O(k^2 \log k)$ polynomial $P(z)=\sum_{j=0}^d c(j) \cdot z^j$ satisfying
\begin{enumerate}
\item $P(z_i)=0$ for each $i \in [k]$.
\item Coefficients $c(0)=\Omega(1)$, $c(j)=O(1)$ and $\sum_{j=1}^d |c(j)|^2 =O(k) \cdot |c(0)|^2$.
\end{enumerate}
\end{theorem}
\begin{corollary}\label{cor:linear_comb}
Given any $g(t)=\sum_{j=1}^k v_j e^{2 \pi \bi f_j t}$ and $\theta>0$, there exist $d=O(k^2 \log k)$ and a sequence of coefficients $(\alpha_1,\ldots,\alpha_d)$ such that
\begin{enumerate}
\item $\alpha_j=O(1)$ for any $j=1,\ldots,d$.
\item For any $x$ (not necessarily in $[-1,1]$), $g(x)=\sum_{j=1}^d \alpha_j \cdot g(x-j \theta)$.
\end{enumerate}
\end{corollary}
\begin{proof}
Given $\theta$, we set $z_i=e^{-2 \pi \bi f_j \theta}$ and apply Theorem~\ref{thm:sup_sup_coef} to obtain coefficients $c(0),\ldots,c(d)$. Then we set $\alpha_j=-c(j)/c(0)$. It is straightforward to verify the second property because of 
$$e^{2 \pi \bi f_j x} - \sum_{j} \alpha_j \cdot e^{2 \pi \bi f_j (x - j \theta)}=0.$$
\end{proof}

The proof of Theorem~\ref{thm:sup_sup_coef} requires the following bound on the coefficients of residual polynomials, which is stated as Lemma 5.3 in \cite{CKPS17}.
\begin{lemma}\label{lem:bound_z_n}
Given $z_1,\ldots,z_k$, for any integer $n$, let $r_{n,k}(z)=\sum_{i=0}^{k-1} r_{n,k}^{(i)} \cdot z^i$ denote the residual polynomial of $r_{n,k} \equiv z^n \mod \prod_{j=1}^k (z-z_j)$. Then each coefficient in $r_{n,k}$ is bounded: $|r_{n,k}^{(i)}| \le {k-1 \choose i}\cdot {n \choose k-1}$ for $n \ge k$ and $|r_{n,k}^{(i)}| \le {k-1 \choose i}\cdot {|n|+k-1 \choose k-1}$ for $n < 0$. 
\end{lemma}

We finish the proof of Theorem~\ref{thm:sup_sup_coef} here.

\begin{proof}
Let $C_0$ be a large constant and $d=5 \cdot k^2 \log k$. We use $\mathcal{P}$ to denote the following subset of polynomials with bounded coefficients: 
$$
\left\{ \sum_{j=0}^d \alpha_j \cdot 2^{-j/k} \cdot z^j \bigg| \alpha_0,\ldots,\alpha_d \in [-C_0,C_0] \cap \mathbb{Z} \right\}.
$$
For each polynomial $P(z) \in \mathcal{P}$, we rewrite $P(z) \mod \prod_{j=1}^k (z-z_j)$ as
$$
\sum_{j=0}^d \alpha_j \cdot 2^{-j/k} \cdot \left(z^j \mod \prod_{j=1}^k (z-z_j) \right)=\sum_{i=0}^{k-1} \left(\sum_{j=0}^d \alpha_j \cdot 2^{-j/k} \cdot r_{n,k}^{(i)} \right) z^i.
$$
The coefficient $\sum_{j=0}^d \alpha_j\cdot 2^{-j/k} \cdot r_{n,k}^{(i)}$ is bounded by
\[
\sum_{j=0}^d C_0 \cdot 2^{-j/k} \cdot 2^k j^{k-1} \le d \cdot C_0 \cdot 2^k \cdot d^{k} \le d^{2k}. %\cdot \max_{j=0,\ldots,d} \left\{ \frac{j^{k-1}}{2^j} \right\}= d \cdot C_0 \cdot 2^k \cdot (3k)^k=2^{2k \log k}.
\]
Then we apply the pigeonhole principle on the $(2C_0+1)^d$ polynomials in $\mathcal{P}$ after module $\prod_{j=1}^d (z-z_j)$: there exist $m > (2C_0+1)^{0.9d}$ polynomials $P_1,\ldots,P_m$ such that each coefficient of $(P_i-P_j) \mod \prod_{j=1}^k (z-z_j)$ is $d^{-2k}$ small from the counting 
\[
\frac{(2C_0+1)^d}{(d^{2k}/4d^{-2k})^{k}} > (2C_0+1)^{0.9d}.
\]
Because $m > (2C_0+1)^{0.9d}$, there exists $j_1 \in [m]$ and $j_2 \in [m]\setminus\{j_1\}$ such that the lowest monomial $z^l$ with different coefficients in $P_{j_1}$ and $P_{j_2}$ satisfies $l \le 0.1d$. Eventually we set 
$$
P(z)=z^{-l} \cdot \big(P_{j_1}(z)-P_{j_2}(z) \big) - \bigg(z^{-l} \mod \prod_{j=1}^k (z-z_j) \bigg) \cdot \bigg(P_{j_1}(z)-P_{j_2}(z) \mod \prod_{j=1}^k (z-z_j) \bigg)
$$ 
to satisfy the first property $P(z_1)=P(z_2)=\cdots=P(z_k)=0$. We prove the second property in the rest of this proof.

We bound every coefficient in $\big(z^{-l} \mod \prod_{j=1}^k (z-z_j) \big) \cdot \big(P_{j_1}(z)-P_{j_2}(z) \mod \prod_{j=1}^k (z-z_j) \big)$ by 
$$
k \cdot \text{max-coefficient}\bigg(z^{-l} \mod \prod_{j=1}^k (z-z_j)\bigg) \cdot \text{max-coefficient}\bigg(P_{j_1}(z)-P_{j_2}(z) \mod \prod_{j=1}^k (z-z_j) \bigg),
$$
which is less than $k \cdot 2^k (l+k)^{k-1} \cdot d^{-2k} \le d \cdot 2^k d^{k-1} \cdot d^{-2k} \le d^{-0.5k}$ from Lemma~\ref{lem:bound_z_n} and the above discussion. 

On the other hand, the constant coefficient in $z^{-l} \cdot \big(P_{j_1}(z)-P_{j_2}(z) \big)$ is at least $2^{-l/k} \ge 2^{-0.1d/k}=k^{-0.5k}$ because $z^l$ is the smallest monomial with different coefficients in $P_{j_1}$ and $P_{j_2}$ from $\mathcal{P}$. Thus the constant coefficient $|C(0)|^2$ of $P(z)$ is at least $0.5 \cdot 2^{-2l/k}$.

Next we upper bound the sum of the rest of the coefficients $\sum_{j=1}^d |C(j)|^2$ by 
$$
\sum_{j=1}^d (2C_0 \cdot 2^{-(l+j)/k} + d^{-0.5k})^2 \le 2 \cdot 4C_0^2 \sum_{j=1}^d 2^{-2(l+j)/k} + 2 \cdot \sum_{j=1}^d d^{-0.5k \cdot 2} \lesssim k \cdot 2^{-2l/k}, 
$$
which demonstrates the second property after normalizing $C(0)$ to 1.
\end{proof}

%% file: append.tex
\section{Properties of the Filter function}\label{appen:filter}
We show basic properties of our filter function in Appendix~\ref{sec:properties_filter} and prove Theorem~\ref{thm:HwithScaling} in Appendix~\ref{sec:proof_filter}.

\subsection[]{Properties of $H$}\label{sec:properties_filter}
We use two bounds on the $\sinc$ function:
\begin{enumerate}
\item For any $|x| \ge \frac{1.2}{\pi}$, $\sinc(x) \le \frac{1}{\pi |x|}$.
\item For any $|x| \le \frac{1.2}{\pi}$, $\sinc(x) \in [1 - \frac{\pi^2 |x|^2}{6}, 1 - \frac{\pi^2 |x|^2}{10}]$.
\end{enumerate}

Without loss of generality, we assume both $R$ and $C$ are powers of 2 and $R \ge S$ (otherwise set $R=S$). Recall that $C$ is even in this section. 

We use $g(t)$ to denote the product of sinc functions in $H(t)$ for convenience:
$$
g(t)=\left( \sinc(C R \cdot t)^{C \log R} \cdot \sinc\big(C \cdot S \cdot t\big)^{C} \cdot \sinc\big(\frac{C \cdot S}{2} \cdot t\big)^{2C} \cdots \sinc\big(C \cdot t\big)^{C \cdot S} \right)
$$

We fix $l=\log_2 (S)$ in this section and rewrite $g(t)$ as 
$$
\sinc(C R \cdot t)^{C \log R} \cdot \prod_{j=0}^{l} \sinc\big(2^{-j} \cdot C \cdot S \cdot t\big)^{2^j \cdot C}.
$$ Before we show the properties of $H$, we consider the tail of $g(t)$.

\begin{claim}\label{clm:bounds_g}
\begin{enumerate}
\item $g(t)=\Theta(1)$ for $|t| \le \frac{1.2}{\pi CR \cdot \sqrt{C \log R}}$.

\item $g(t)=e^{-\Theta(|CR \cdot t|^2 \log R)}$ for $|t| \in [\frac{1.2}{\pi CR \cdot \sqrt{C \log R}}, \frac{1.2}{\pi CR}]$.

\item $g(t) \le (\frac{1}{\pi \cdot CR \cdot |t|})^{C \log R}$ for $|t| \in [\frac{1.2}{\pi CR}, \frac{1.2}{\pi C \cdot S}]$.

\item For any $i \in [l]$, $g(t) \le (\frac{1}{\pi \cdot CR \cdot |t|})^{C \log R} \cdot 1.2^{-(2^{i+1}-1)C}$ for any $|t| \in [\frac{1.2 \cdot 2^{i-1}}{\pi C \cdot S},\frac{1.2 \cdot 2^{i}}{\pi C \cdot S}]$.

\item $g(t) \le (\frac{1}{\pi CR \cdot t})^{C \log R} \cdot \prod_{j=0}^{l} (\frac{1}{\pi 2^{-j} \cdot C \cdot S \cdot t})^{2^j \cdot C}$ for $|t| \ge \frac{1.2 \cdot 2^{l}}{\pi C \cdot S} = \frac{1.2}{C \pi}$.
\end{enumerate}
\end{claim}
\begin{proof}
We first bound $\sinc(C R \cdot t)^{C \log R}$ then bound $\prod_{j=0}^{l} \sinc\big(2^{-j} \cdot C \cdot S \cdot t\big)^{2^j \cdot C}$.

\begin{enumerate}
\item For $|t| \le \frac{1.2}{\pi CR}$, from the second property of $\sinc$ functions, 
$$
\sinc(C R \cdot t) \in \left[ 1 - \frac{\pi^2 |C R t|^2}{6}, 1 - \frac{\pi^2 |C R t|^2}{10} \right] \Rightarrow 
\sinc(C R \cdot t)^{C \log R}=\Theta(1) \text{ for } |t| \le \frac{1.2}{\pi CR \cdot \sqrt{C \log R}}
$$
and 
$$\sinc(C R \cdot t)^{C \log R}=e^{-\Theta(|CR \cdot t|^2 \log R)} \text{ for } t \in [\frac{1.2}{\pi CR \cdot \sqrt{C \log R}}, \frac{1.2}{\pi CR}].
$$

\item For $|t| \ge \frac{1.2}{\pi CR}$, from the first property of $\sinc$ functions, 

$$\sinc(C R \cdot t)^{C \log R} \le (\frac{1}{\pi \cdot CR \cdot |t|})^{C \log R}.$$
\end{enumerate}
Then we bound the tail of the product of sinc functions.

\begin{enumerate}
\item For $|t| \le \frac{1.2}{\pi C \cdot S}$, 
$$
\sinc\big(2^{-j} \cdot C \cdot S \cdot t\big)^{2^j \cdot C} \in \left[ \big(1 - \frac{\pi^2 \cdot |2^{-j} \cdot C \cdot S \cdot t|^2}{6} \big)^{2^j \cdot C}, \big( 1 - \frac{\pi^2 \cdot |2^{-j} \cdot C \cdot S \cdot t|^2}{10} \big)^{2^j \cdot C} \right].
$$ 
Notice that  $\pi^2 \cdot |2^{-j} \cdot C \cdot S \cdot t|^2$ is less than $1.2^2 \cdot 2^{-2j}$. Thus $\sinc\big(2^{-j} \cdot C \cdot S \cdot t\big)^{2^j \cdot C} = \big( 1 - \Theta(2^{-j}) \big)^C$ and their products over $j$ is 
$$\prod_{j=0}^{l} \sinc\big(2^{-j} \cdot C \cdot S \cdot t\big)^{2^j \cdot C} = \left( 1- \Theta\left( \sum_{j=0}^l 2^{-j} \right) \right)^C=\Theta(1)^C = \Theta(1).$$

\item Let us fix $i\le l$ and consider $\sinc\big(2^{-j} \cdot C \cdot S \cdot t\big)^{2^j \cdot C}$ for $|t| \in [\frac{1.2 \cdot 2^{i-1}}{\pi C \cdot S},\frac{1.2 \cdot 2^{i}}{\pi C \cdot S}]$. By the first property of sinc function, for $j \le i$,
$$
\sinc\big(2^{-j} \cdot C \cdot S \cdot t\big)^{2^j \cdot C} \le (\frac{1}{\pi \cdot 2^{-j} \cdot C \cdot S \cdot |t|})^{2^j \cdot C} \le (\frac{1}{1.2 \cdot 2^{-j+i}})^{2^j \cdot C} \le 1.2^{- 2^j \cdot C}.
$$
For $j>i$, we use the same analysis with the second property of the sinc function:
$$
\sinc\big(2^{-j} \cdot C \cdot S \cdot t\big)^{2^j \cdot C} \in \left[ \big(1 - \frac{\pi^2 \cdot |2^{-j} \cdot C \cdot S \cdot t|^2}{6} \big)^{2^j \cdot C}, \big( 1 - \frac{\pi^2 \cdot |2^{-j} \cdot C \cdot S \cdot t|^2}{10} \big)^{2^j \cdot C} \right]
$$ where $\pi^2 \cdot |2^{-j} \cdot C \cdot S \cdot t|^2$ is at least $1.2^2 \cdot 2^{-2(j-i)}$. Hence the product is
$$
\prod_{j=0}^{l} \sinc\big(2^{-j} \cdot C \cdot S \cdot t\big)^{2^j \cdot C} \le 1.2^{ - \sum_{j=0}^i 2^j \cdot C} \cdot \prod_{j=i+1}^l \left(1 - \frac{1.2^2 \cdot 2^{-2(j-i)}}{6} \right)^{2^j \cdot C} \le 1.2^{-(2^{i+1}-1)C}.
$$
\end{enumerate}
We get the tail bounds by combining the above discussion of $\sinc(C R \cdot t)^{C \log R}$ and $\prod_{j=0}^{l} \sinc\big(2^{-j} \cdot C \cdot S \cdot t\big)^{2^j \cdot C}$ together.
\end{proof}

Since $H(t)=s_0 \cdot g(t) * \rect_2(t) = s_0 \cdot \int_{t-1/2}^{t+1/2} g(x) \mathrm{d} x$, we have the following bounds on $H(t)$ based on Claim~\ref{clm:bounds_g}. 

\begin{lemma}\label{lem:bounds_H}
For any constant $C \ge 4$, 
\begin{enumerate}
\item $s_0=\Theta(\pi CR \cdot \sqrt{C \log R})$.
\item $H(t) = 1 \pm 0.01$ for $|t| \le \frac{1}{2} - \frac{1.2}{\pi C R}$.
\item $H(t) \lesssim \frac{s_0}{S} \cdot R^{-C/4}$ for $|t| \in [\frac{1}{2} + \frac{1.2}{\pi C R}, \frac{1}{2} + \frac{1.2}{\pi C \cdot S}]$.
\item $H(t) \lesssim s_0 \cdot R^{-C/4} \cdot 1.2^{-2^i C}$ for $|t| \in [\frac{1}{2} + \frac{1.2 \cdot 2^{i-1}}{\pi C \cdot S}, \frac{1}{2} + \frac{1.2 \cdot 2^{i}}{\pi C \cdot S}]$ of any $i \le [l]$.
\item $H(t) \le s_0 \cdot (\frac{1}{1.2\pi CR \cdot (|t|-\frac{1}{2})})^{C \log R} \cdot \big(\frac{1}{C \pi \cdot (|t|-\frac{1}{2})}\big)^{C S}$ for $t \ge \frac{1}{2} + \frac{1.2}{C \pi}$.
\end{enumerate}
\end{lemma}

\begin{proof}
We bound the integration of different intervals of $g(t)$ as follows:
\begin{enumerate}
\item $\int_{\frac{-1.2}{\pi CR}}^{\frac{1.2}{\pi CR}} g(x) \mathrm{d} x = \int_{\frac{-1.2}{\pi CR \cdot \sqrt{C \log R}}}^{\frac{1.2}{\pi CR \cdot \sqrt{C \log R}}} g(x) \mathrm{d} x + 2\int_{\frac{1.2}{\pi CR \cdot \sqrt{C \log R}}}^{\frac{1.2}{\pi CR}} e^{-\Theta(|CR \cdot x|^2 \log R)} \mathrm{d} x= \Theta(\frac{1}{\pi CR \cdot \sqrt{C \log R}})$.

\item $\int_{\frac{1.2}{\pi CR}}^{\frac{1.2}{\pi C \cdot S}} g(x) \mathrm{d} x \le \int_{\frac{1.2}{\pi CR}}^{\frac{1.2}{\pi C \cdot S}} (\frac{1}{\pi \cdot CR \cdot x})^{C \log R} \mathrm{d} x \le \frac{1.2}{\pi C \cdot S} \cdot 1.2^{-C \log R}$.

\item For any $i \in [l]$ of $l=\log_2 S$, 
\begin{align*}
\int_{\frac{1.2 \cdot 2^{i-1}}{\pi C \cdot S}}^{\frac{1.2 \cdot 2^{i}}{\pi C \cdot S}} g(x) \mathrm{d} x & \le \int_{\frac{1.2 \cdot 2^{i-1}}{\pi C \cdot S}}^{\frac{1.2 \cdot 2^{i}}{\pi C \cdot S}} (\frac{1}{\pi \cdot CR \cdot x})^{C \log R} \cdot 1.2^{-(2^{i+1}-1)C} \mathrm{d} x \\
& \le \frac{1.2 \cdot 2^{i}}{\pi C \cdot S} \cdot (\frac{S}{1.2 \cdot 2^{i-1} R})^{C \log R} \cdot 1.2^{-(2^{i+1}-1)C} \\
& \le \frac{1.2 \cdot 2^{i}}{\pi C \cdot S} \cdot R^{-C/4} \cdot 1.2^{-2^{i} C}.
\end{align*}

\item For $|t| \ge \frac{1.2}{C \pi}$, 
\begin{align*}
\int_{t}^{t+1} g(x) \mathrm{d} x & \le \int_{t}^{t+1} (\frac{1}{\pi CR \cdot x})^{C \log R} \cdot \prod_{j=0}^{l} (\frac{1}{\pi 2^{-j} \cdot C \cdot S \cdot x})^{2^j \cdot C} \mathrm{d} x \\
& \le (\frac{1}{\pi CR \cdot t})^{C \log R} \cdot (\frac{1}{\pi C \cdot t})^{2^l \cdot C} \\
& \le (\frac{1}{\pi CR \cdot t})^{C \log R} \cdot (\frac{1}{\pi C \cdot t})^{C S}.
\end{align*}
\end{enumerate}

Next we prove all claims in this lemma.
\begin{enumerate}
\item For $s_0$, notice that 
$$
\int_{-1/2}^{1/2} g(x) \mathrm{d} x \le \int_{\frac{-1.2}{\pi CR}}^{\frac{1.2}{\pi CR}} g(x) \mathrm{d} x + \int_{|x| \in (\frac{1.2}{\pi CR},1/2]} g(x) \mathrm{d} x=\Theta(\frac{1}{\pi CR \cdot \sqrt{C \log R}}) + O(\frac{1.2}{\pi C \cdot S} \cdot 1.2^{-C \log R}),
$$ which also indicates $s_0 \in [1,1+10^{-3}] \cdot 1/\left(\int_{\frac{-1.2}{\pi CR}}^{\frac{1.2}{\pi CR}} g(x) \mathrm{d} x \right)$.

\item When $|t|<\frac{1}{2} - \frac{1.2}{\pi C R}$, 
$H(t) = s_0 \cdot \left( \int_{\frac{-1.2}{\pi CR}}^{\frac{1.2}{\pi CR}} g(x) \mathrm{d} x + \int_{[t-1/2,t+1/2] \setminus [\frac{-1.2}{\pi CR},\frac{1.2}{\pi CR}]} g(x) \mathrm{d} x \right)$, which is in 
$s_0 \cdot [1,1+10^{-3}] \cdot \int_{\frac{-1.2}{\pi CR}}^{\frac{1.2}{\pi CR}} g(x) \mathrm{d} x \subseteq [1 - 0.01, 1+ 0.01].$

\item When $|t| \in [\frac{1}{2} - \frac{1.2}{\pi C R}, \frac{1}{2} + \frac{1.2}{\pi C R}]$, $H(t) \in [0,1]$.

\item When $|t| \in [\frac{1}{2} + \frac{1.2}{\pi C R}, \frac{1}{2} + \frac{1.2}{\pi C \cdot S}]$, 
$$
H(t) \le s_0 \cdot \left( \int_{\frac{1.2}{\pi CR}}^{\frac{1.2}{\pi C \cdot S}} g(x) \mathrm{d} x + \sum_{j=1}^l \int_{\frac{1.2 \cdot 2^{j-1}}{\pi C \cdot S}}^{\frac{1.2 \cdot 2^{j}}{\pi C \cdot S}} g(x) \mathrm{d} x + \int_{\frac{1.2}{\pi}}^{\frac{1.2}{\pi} + 1} g(x) \mathrm{d} x\right) \le 2 s_0 \cdot  \int_{\frac{1.2}{\pi CR}}^{\frac{1.2}{\pi C \cdot S}} g(x) \mathrm{d} x.
$$

\item When $|t| \in [\frac{1}{2}+\frac{1.2 \cdot 2^{i-1}}{\pi C \cdot S}, \frac{1}{2}+\frac{1.2 \cdot 2^{i}}{\pi C \cdot S}]$ of a positive integer $i<l$, 
$$
H(t) \le s_0 \cdot \left( \sum_{j=i}^l \int_{\frac{1.2 \cdot 2^{j-1}}{\pi C \cdot S}}^{\frac{1.2 \cdot 2^{j}}{\pi C \cdot S}} g(x) \mathrm{d} x + \int_{\frac{1.2}{\pi}}^{\frac{1.2}{\pi} + 1} g(x) \mathrm{d} x\right) \le 2s_0 \cdot \frac{1.2 }{C \pi} \cdot R^{-C/4} \cdot 1.2^{-2^i C}.
$$

\item When $t>\frac{1}{2}+\frac{1.2}{C \pi}$, we use the bound in the last item of the above discussion.
\end{enumerate}
\end{proof}

\subsection{Proof of Theorem~\ref{thm:HwithScaling}}\label{sec:proof_filter}
We finish the proof of Theorem~\ref{thm:HwithScaling} using Lemma~\ref{lem:bounds_H} for $\alpha=\frac{1}{2} + \frac{1.2}{\pi C R}$. Without loss of generality, we assume $R \ge S$ in this proof (otherwise set $R=S$).

We first show
$$
\int_{-1}^1 |g(x) \cdot H\big( \alpha x\big)|^2 \mathrm{d} x \ge 0.9 \int_{-1}^1 |g(x)|^2 \mathrm{d} x.
$$
From the second property of $H$ in Lemma~\ref{lem:bounds_H}, $H\big( \alpha x\big) \ge 1 - 0.01$ for any $|x| \le \frac{\frac{1}{2} - \frac{1.2}{\pi C R}}{\alpha} = 1 - \frac{2.4}{\pi C R + 2.4}$ such that
$$\int_{-1 + \frac{2.4}{\pi C R/2 + 1.2}}^{1 - \frac{2.4}{\pi C R/2 + 1.2}} |g(x) \cdot H(\alpha x)|^2 \mathrm{d} x \ge 0.99^2 \int_{-1 + \frac{2.4}{\pi C R/2 + 1.2}}^{1 - \frac{2.4}{\pi C R/2 + 1.2}} |g(x)|^2 \mathrm{d} x.$$ 

At the same time, $|g(t)|^2 \le R \cdot \underset{x \sim [-1,1]}{\E}[|g(x)|^2] = R/2 \cdot \int_{-1}^1 |g(x)|^2 \mathrm{d} x$ for any $t \in [-1,1]$. This indicates
$$\int_{-1 + \frac{2.4}{\pi C R + 2.4}}^{1 - \frac{2.4}{\pi C R + 2.4}} |g(x)|^2 \mathrm{d} x \ge (1 -  \frac{R/2 \cdot 2.4}{\pi C R + 2.4}) \int_{-1}^1 |g(x)|^2 \mathrm{d} x.
$$
The first property follows from these two inequalities.

In the rest of this proof, we apply Lemma~\ref{lem:bounds_H} to prove:
$$
\int_{-\infty}^{-1} |g(x) \cdot H\big( \alpha x\big)|^2 \mathrm{d} x + \int_{1}^{\infty} |g(x) \cdot H\big( \alpha x\big)|^2 \mathrm{d} x\le 0.04 \int_{-1}^1 |g(x)|^2 \mathrm{d} x.
$$
We split $\int_{1}^{\infty} |g(x) \cdot H\big( \alpha x\big)|^2 \mathrm{d} x$ into several intervals:
$$
\int_{1}^{(\frac{1}{2}+\frac{1.2}{\pi C \cdot S})/\alpha} |g(x) \cdot H\big( \alpha x\big)|^2 \mathrm{d} x + \sum_{i=1}^{\log_2 S} \int_{(\frac{1}{2} + \frac{1.2 \cdot 2^{i-1}}{\pi C \cdot S})/\alpha}^{(\frac{1}{2} + \frac{1.2 \cdot 2^{i}}{\pi C \cdot S})/\alpha} |g(x) \cdot H\big( \alpha x\big)|^2 \mathrm{d} x + \int_{(\frac{1}{2} + \frac{1.2}{\pi C})/\alpha}^{\infty} |g(x) \cdot H\big( \alpha x\big)|^2 \mathrm{d} x.
$$
In the first two terms, we rewrite $|g(t)| \le \poly(R) \cdot \|g\|_2 \cdot t^S$ as $\poly(R) \cdot \|g\|_2 \cdot e^{(t-1)S}$. By the third and fourth properties of $H(t)$ in Lemma~\ref{lem:bounds_H}, their summations is less than $0.01 \|g\|_2^2$. For the last term, given the last property of $H(t)$ in Lemma~\ref{lem:bounds_H} and a large constant $C$, we have $$H(\alpha t) \le s_0 \cdot (\frac{1}{1.2 R})^{C \log R} \cdot (\frac{1}{2t})^S \text{ when } t \ge (\frac{1}{2} + \frac{1.2}{\pi C})/\alpha.$$ It is straightforward to verify that $\int_{1}^{\infty} |g(x) \cdot H\big( \alpha x\big)|^2 \mathrm{d} x \le 0.02 \cdot \|g\|^2_2$.

The last property follows from the upper bounds in Lemma~\ref{lem:bounds_H}.

%% file: append2.tex
\section{Omitted Proofs in Section~\ref{sec:condition_num_growth}}\label{append:sparse_signal}
We first prove Theorem~\ref{thm:sparse_freq_estimation} then finish the proof of Theorem~\ref{thm:worst_case_sFFT} and \ref{thm:bound_growth_around1} in Appendix~\ref{sec:proof_R} and \ref{sec:growth} separately.

\subsection{Proof of Theorem~\ref{thm:sparse_freq_estimation}}\label{sec:proof_sparse_freq}
We finish the proof of Theorem~\ref{thm:sparse_freq_estimation} in this section. The only difference compared to Theorem~\ref{thm:general_bound} is to use a biased distribution $D$ such that we could improve the sample complexity to $\wt{O}(k \log \frac{F}{\Delta' \epsilon})$.

\paragraph{How to Generate Samples.} We will use a distribution $D$ not uniform on $[-1,1]$ to generate the random samples. For $m$ samples $x_1,\cdots,x_m \sim D$, we assign a weight $w_i=\frac{1}{2m \cdot D(x_i)}$ for each sample $x_i$ such that for any function $h$,
$$
\E_{x_1,\cdots,x_m \sim D}\bigg[\sum_{i=1}^m w_i |h(x_i)|^2 \bigg]=\sum_{i=1}^m \E_{x_i 
\sim D}\bigg[ \frac{1}{2m \cdot D(x_i)} |h(x_i)|^2 \bigg] = \sum_{i=1}^m \E_{x 
\sim [-1,1]}\bigg[ \frac{1}{m} |h(x_i)|^2 \bigg]=\|h\|_2^2.
$$

\cite{CP18} presented an explicit distribution $D$ such that $\tilde{O}(k)$ samples could guarantee $\sum_{i=1}^m w_i |g(x_i)|^2 $ is close to $\|g\|_2^2$ with high probability. For completeness, we show it with our improved bound $R$.
\begin{lemma}\label{lem:generate_sampling}
Given the sparsity $k$, there exists a constant $c$ such that the distribution 
\begin{flalign*}
& D_{\FF}(x)=
\begin{cases}
\frac{c}{(1-|x|) \log k }, & \text{ for } |x| \le 1-\frac{1}{k^2 \log^2 k}\\
c \cdot k^2 \log k, & \text{ for } |x| > 1-\frac{1}{k^2 \log^2 k}
\end{cases} 
\end{flalign*}
guarantees for any $k$-Fourier-sparse signal $g$, $\underset{x \in [-1,1]}{\sup} \frac{1}{2 D(x)} \cdot \frac{|g(x)|^2}{\|g\|_2^2} = O(k \log^2 k)$.

Moreover, $m=O(\frac{k \log^2 k \log \frac{1}{\delta}}{\epsilon^2})$ samples $x_1,\cdots,x_m$ from $D$ with weights $w_i=\frac{1}{2m \cdot D(x_i)}$ for $i \in [m]$ guarantee that, with probability at least $1-\delta$,
$$
\sum_{i=1}^m w_i \cdot |g(x_i)|^2 \in [1\pm \epsilon] \cdot \|g\|_2^2.
$$
\end{lemma}
\begin{proof}
Given $D$ and the $k$-Fourier-sparse signal $g$, let $z(x)$ denote $\frac{|g(x)|^2}{2 D(x)}$ for $x \in [-1,1]$. We have $\E_{x \sim D}\big[ z(x) \big]=\E_{x \sim [-1,1]}\big[ |g(x)|^2 \big]=\|g\|_2^2$ and $\underset{x \in \supp(D)}{\sup} \frac{z(x)}{\E_{x' \sim D}[z(x')]} = O( k \log^2 k)$. We apply the Chernoff bound in Lemma~\ref{lem:chernoff_bound} on the random variables $z(x_1),\cdots,z(x_m)$ to obtain the statement.
\end{proof}

Similar to Lemma~\ref{lem:good_sampling}, we state the following version for Fourier-sparse signals.

\begin{lemma}\label{lem:good_sampling_sparse}
Given the sparsity $k, f_0$  and $\Delta$, let $g$ be a $k$-Fourier-sparse signal $g(t)=\sum_{i \in [k]} v_i \cdot e^{2 \pi \i f_i t}$ with $f_i \subseteq [f_0-\Delta,f_0+\Delta]$ and $\Delta'=\Delta + O(\frac{R \log k + k^2 \log^2 k}{T})$.
    
Let $y(t)=g(t)+\eta(t)$ be the observable signal on $[-1,1]$ where the noise $\|\eta\|^2_2 \le \epsilon \|g\|^2_2$ for a sufficiently small constant $\epsilon$. There exist a constant $\gamma$ and an algorithm such that for any $\beta \le \frac{\gamma}{\Delta'}$, it takes $O(k \log^2 k)$ samples to output $\alpha$ satisfying $|y_H(\alpha) e^{2 \pi i f_0 \beta}-y_H(\alpha+\beta)| \le 0.3 |y_H(\alpha)|$ with probability at least 0.6.
\end{lemma}

We show our algorithm in Algorithm~\ref{alg:freq_est_sparse}. We finish the proof of Theorem~\ref{thm:sparse_freq_estimation}.

\begin{algorithm}
\caption{Obtain one good $\alpha$}\label{alg:freq_est_sparse}
\begin{algorithmic}[1]
\Procedure{\textsc{ObtainOneGoodSample}}{$k,y(t)$}
\State Let $m=C \cdot k \log^2 k$ for a large constant $C$.
\State Take $m$ samples $x_1,\cdots,x_m$ from the distribution $D$ in Lemma~\ref{lem:generate_sampling}.
\State Assign a weight $w_i=\frac{1}{2 m \cdot D(x_i)}$ for each sample $x_i$.
\State Query $y(x_i)$ and compute $y_H(x_i)=y(x_i) \cdot H(x_i)$ for each $i$.
\State Set a distribution $D_m$ proportional to $w_i \cdot |y_H(x_i)|^2$, i.e., $D_m(x_i)=\frac{w_i \cdot |y_H(x_i)|^2}{\sum_{j=1}^m w_j \cdot |y_H(x_j)|^2}$.
\State Output $\alpha \sim D_m$.
\EndProcedure
\end{algorithmic}
\end{algorithm}

\begin{proofof}{Theorem~\ref{thm:restate_sparse}}
From Lemma~\ref{lem:good_sampling_sparse}, $\frac{y_H(\alpha+\beta)}{y_H(\alpha)}$ gives a good estimation of $e^{2 \pi i f_0 \beta}$ with probability 0.6 for any $\beta \le \frac{\gamma}{\Delta'}$. We use the frequency search algorithm of Lemma 7.3 in \cite{CKPS17} with the sampling procedure in Lemma~\ref{lem:good_sampling_sparse}. Because the algorithm in \cite{CKPS17} uses the sampling procedure $O(\log \frac{F}{\Delta' \cdot \delta})$ times to return a frequency $\wt{f}$ satisfying $|\wt{f} - f_0| \le \Delta'$ with prob. at least $1-\delta$, the sample complexity is $O(k \log^2 k \cdot \log \frac{F}{\Delta' \cdot \delta})$.
\end{proofof}

\subsection{Proof of Theorem~\ref{thm:worst_case_sFFT}}\label{sec:proof_R}
We bound $R$ of $k$-sparse-Fourier signals in this section. We first state the technical result to prove the upper bound $R$.
\begin{theorem}\label{thm:sFFT_interpolation}
Given any $k>0$, there exists $d=O(k^2 \log k)$ such that for any $g(x)=\sum_{j=1}^k v_j \cdot e^{2 \pi \bi f_j \cdot x}$,  any $t \in \mathbb{R}$, and any $\theta>0$,
$$
|g(t)|^2 \le O(k) \cdot \bigg(\sum_{j=1}^{d} |g(t+j \cdot \theta|^2 \bigg).
$$
\end{theorem}

\begin{proofof}{Theorem~\ref{thm:sFFT_interpolation}}
Given $k$ frequencies $f_1,\cdots,f_k$ and $\theta$, we set $z_1=e^{2 \pi \bi f_1 \cdot \theta},\cdots,z_k=e^{2 \pi \bi f_k \cdot \theta}$. Let $C(0), \cdots, C(d)$ be the coefficients of the degree $d$ polynomial $P(z)$ in Theorem~\ref{thm:sup_sup_coef}. We have
\begin{align*}
\sum_{j=0}^d C(j) \cdot g(t+j \cdot \theta) & =\sum_{j=0}^d C(j) \sum_{j' \in [k]} v_{j'} \cdot e^{2 \pi \bi f_{j'} (t+j \theta)}\\ 
& =\sum_{j=0}^d C(j) \sum_{j' \in [k]} v_{j'} \cdot e^{2 \pi \bi f_{j'} t} \cdot z_{j'}^j = \sum_{j' \in [k]} v_{j'} \cdot e^{2 \pi \bi f_{j'} t} \sum_{j=0}^d C(j) \cdot z_{j'}^j=0.
\end{align*}
Hence for every $i \in [k]$,
\begin{equation}\label{eq:linear_relation_0}
-C(0) \cdot g(t)=\sum_{j=1}^{d} C(j) \cdot g(t+j \cdot \theta).
\end{equation}
By Cauchy-Schwartz inequality, we have
\begin{equation}\label{eq:linear_relation_head}
|C(0)|^2 \cdot |g(t)|^2 \le \left(\sum_{j=1}^{d} |C(j)|^2\right) \cdot \left(\sum_{j=1}^{d} |g(t+j \cdot \theta)|^2 \right).
\end{equation}
From the second property of $C(0),\cdots,C(d)$ in Theorem~\ref{thm:sup_sup_coef}, $|g(t)|^2 \le O(k) \cdot \bigg(\sum_{j=1}^{d} |g(t+j \cdot \theta|^2 \bigg)$.
\end{proofof}

We finish the proof of Theorem~\ref{thm:worst_case_sFFT} bounding $R$ by the above relation. For convenience, we restate it for $T=1$.

\begin{theorem}
For any $g(t)=\sum_{j=1}^k v_j e^{2 \pi \bi f_j t}$, 
\[
\frac{\underset{x \in [-1,1]}{\sup} |g(x)|^2}{\underset{x \in [-1,1]}{\E}[|g(x)|^2]} = O(k^3 \log^2 k).
\]
\end{theorem}

\begin{proof}
We prove
\[
|g(t)|^2 = O(k^3 \log^2 k) \int_{t}^1 |g(x)|^2 \mathrm{d} x \text{ for any } t \le 0,
\] 
which indicates $|g(t)|^2 = O(k^3 \log^2 k) \cdot \underset{x \sim [-1,1]}{\E}\big[ |g(x)|^2 \big]$. By symmetry, it also implies that $|g(t)|^2 = O(k^3 \log^2 k) \cdot \underset{x \sim [-1,1]}{\E}\big[ |g(x)|^2 \big]$ for any $t \ge 0$.

We use Theorem~\ref{thm:sFFT_interpolation} on $g(t)$:
\begin{align*}
\frac{1-t}{d} \cdot |g(t)|^2 & \le O(k) \cdot \int_{\theta=0}^{\frac{1-t}{d}} \sum_{j \in [d]}  |g(t+j \theta)|^2 \mathrm{d} \theta\\
& \lesssim k \sum_{j \in [d]} \int_{\theta=0}^{\frac{1-t}{d}} |g(t+j \theta)|^2 \mathrm{d} \theta \\
& \lesssim k \sum_{j \in [d]} \frac{1}{j} \cdot \int_{\theta'=0}^{\frac{(1-t) j}{d}} |g(t+\theta')|^2 \mathrm{d} \theta'\\
& \lesssim k \sum_{j \in [d]} \frac{1}{j} \cdot \int_{x=-1}^{1} |g(x)|^2 \mathrm{d} x\\
& \lesssim k \log k \cdot \int_{x=-1}^{1} |g(x)|^2 \mathrm{d} x.
\end{align*}
From all discussion above, we have $|g(t)|^2 \lesssim dk \log k \cdot \underset{x \in [-1,1]}{\E}[|g(x)|^2]$.
\end{proof}

%%%%%%%%%%%%%%%%%%%%%%%%%%%%%%%%%%%%%%%%%%%%%%%%%%%%%%%%%%%%%%%%%%%%%%%%%%%%%%%%%%%%%%%%%%%%%%%%%%%%%%%%%%%%%%%%%%%%%%%%%%%%%%%%%%%%

\subsection{Growth outside of the observation}\label{sec:growth}
We prove Theorem~\ref{thm:bound_growth_around1} which bounds $S=\tilde{O}(k^2)$ in this section. We divide the proof into two parts for $|x| \le 1+1/k$ and $|x| > 1+ 1/k$ separately after fixing $T=1$.

\begin{lemma}\label{lem:bound_close1}
For any $g(t)=\sum_{j=1}^k v_j \cdot e^{2 \pi \bi f_j t}$, there exists a constant $C_1$ such that for any $x \ge 1$, $|g(x)| \lesssim \poly(k) \cdot \|g\|_2 \cdot C_1^{(x-1) \cdot k^2 \log k}$.
\end{lemma}

\begin{remark}
The growth of Chebyshev polynomial is $e^{\Theta(k\sqrt{x-1})}$ for $x=1+O(1/k^2)$.
\end{remark}

\begin{proof}
Let $d=O(k^2 \log k)$ denote the length of the linear combination in Corollary~\ref{cor:linear_comb} and $\theta=\frac{2}{d}$. Given $g(t)$ and $\theta$, we use $\alpha_1,\cdots,\alpha_d$ to denote the coefficients of the linear combination of $g(t),g(t-\theta),\ldots,g(t-d \theta)$ in Corollary~\ref{cor:linear_comb}. For convenience, we use $C_0$ to denote the upper bound on the coefficients $\alpha_j$.

We use induction to prove that for some $C=O(1)$, for any $l$,
\begin{equation}\label{eq:induction}
\text{ for any } x \in (1, 1+\frac{2l}{d}], |g(x)| \le C \cdot d k^{1.5} \log k \cdot \|g\|_2 \cdot (2C_0)^{l}.
\end{equation} 
For base case $l=1$, from Corollary~\ref{cor:linear_comb}, $g(x)=\sum_{j=1}^d \alpha_j \cdot g(x-j \theta)$ where $x - j \theta \in [-1,1]$. Because each $\big|g(x-j \theta)\big| \le C \cdot k^{1.5} \log k \cdot \|g\|_2$ from Theorem~\ref{thm:worst_case_sFFT}, we have $$
\big|g(x)\big| \le \sum_{j=1}^d |\alpha_j| \cdot \big|g(x-j \theta)\big| \le C \cdot C_0 \cdot d \cdot k^{1.5} \log k \cdot \|g\|_2.$$

Suppose \eqref{eq:induction} is true for any $x \in (1, 1 + \frac{2l}{d}]$. Let us consider $x \in (1+\frac{2l}{d}, 1+ \frac{2(l+1)}{d}]$. We still have $g(x)=\sum_{j=1}^d \alpha_j \cdot g(x-j \theta)$ where each $x - j \theta \in (1+\frac{2(l-j)}{d}, 1 + \frac{2(l+1-j)}{d}]$. This indicates 
\begin{align*}
\big|g(x)\big| & \le \sum_{j=1}^d |\alpha_j| \cdot \big|g(x-j \theta)\big| \\
& \le C_0 \sum_{j=1}^d \big|g(x-j \theta)\big| \\
& \le C_0 \sum_{j=1}^l \big|g(x-j \theta)\big| + C_0 \sum_{j=l+1}^d \big|g(x-j \theta)\big| \\
& \le C_0 \sum_{j=1}^l C \cdot d k^{1.5} \log k \cdot \|g\|_2 \cdot (2 C_0)^{l+1-j} + C_0 (d-j) \cdot C \cdot k^{1.5} \log k \cdot \|g\|_2.\\
& \le C_0^{l+1} \cdot C \cdot d k^{1.5} \log k \cdot \|g\|_2 \cdot \sum_{j=1}^l 2^{l+1-j} + C_0 d \cdot C \cdot k^{1.5} \log k \cdot \|g\|_2.\\
& \le C_0^{l+1} \cdot C \cdot d k^{1.5} \log k \cdot \|g\|_2 (2^{l+1} - 2) + C_0 d \cdot C \cdot k^{1.5} \log k \cdot \|g\|_2 \\
& \le C_0^{l+1} \cdot C \cdot d k^{1.5} \log k \cdot \|g\|_2 \cdot 2^{l+1}.
\end{align*}
\end{proof}

For completeness, we bound the growth rate of $|t|>1+1/k$ here, which is a reformulation of Lemma 5.5 in \cite{CKPS17}.
\begin{lemma}\label{lem:bound_growth}
For any $g(t)=\sum_{j=1}^k v_j e^{2 \pi \bi f_j t}$ and any $|t|>1$, 
\[
|g(t)|^2 \lesssim k^3 \cdot (3k \cdot t)^k \cdot \|g\|^2_2.
\]
\end{lemma}

\begin{proof}
We fix $t>1$ in this proof. Let $\theta=1/k$ and $n=\big[(t+1/2)/\theta\big]$ such that $t - n \theta \in [-1/2,-1/2+\theta]$ and $t - (n - k) \theta \in [1/2, 1/2+\theta]$. We first show the coefficients $C_0,\cdots,C_{k-1}$ in 
\[
\sum_{j=0}^{k-1} C_j \cdot z^j = z^n \mod \prod_{j=1}^k (z - e^{2\pi \bi f_j \theta})
\] 
satisfy $g(t)=\sum_{l=0}^{k-1} C_j \cdot g\big( t - (n-l)\theta \big)$. Let $z_j=e^{2\pi \bi f_j \theta}$ such that $z_j^n = \sum_{j=0}^{k-1} C_j \cdot z^{j}$. For $g(t)=\sum_{j=1}^k v_j e^{2 \pi \bi f_j t}$, we rewrite it as 
\begin{align*}
\sum_{j=1}^k v_j e^{2 \pi \bi f_j (t - n \theta)} \cdot e^{2 \pi \bi f_j n \theta} &= \sum_{j=1}^k v_j e^{2 \pi \bi f_j (t - n \theta)} \cdot z_j^n \\
& = \sum_{j=1}^k v_j e^{2 \pi \bi f_j (t - n \theta)} \cdot \sum_{l=0}^{k-1} C_l \cdot z_j^{l} \\
& = \sum_{l=0}^{k-1} C_l \cdot \sum_{j=1}^k v_j e^{2 \pi \bi f_j (t - n \theta)} z_j^{l} \\
& = \sum_{l=0}^{k-1} C_l \cdot g(t - n \theta + l \theta).
\end{align*}
Thus $|g(t)|^2 \le (\sum_{j=0}^{k-1} |C_j|^2) \cdot (\sum_{l=0}^{k-1} |g(t - n \theta + l \theta)|^2)$.

Since $g(t - n \theta + l \theta) \in [- 2/3, 2/3]$, $|g(t - n \theta + l \theta)|^2 \lesssim k \underset{x \in [-1,1]}{\E}[|g(x)|^2]$ from  \cite{CP18}. On the other hand, $|C_j| \le {k-1 \choose j}{n \choose k-1} \le (2n)^{k-1}$ from Lemma~\ref{lem:bound_z_n}.

From all discussion above, 
\[
|g(t)|^2 \lesssim k \cdot (2n)^{k-1} \cdot k^2 \underset{x \in [-1,1]}{\E}[|g(x)|^2] \lesssim k^3 (3kt)^{k} \cdot \underset{x \in [-1,1]}{\E}[|g(x)|^2].
\]
\end{proof}

\begin{proofof}{Theorem~\ref{thm:bound_growth_around1}}
We combine Lemma~\ref{lem:bound_close1} and \ref{lem:bound_growth}:
For $x \le 1+1/k$, $C_1^{(x-1)k^2 \log k} = e^{(x-1) k^2 \log k \log C_1} = x^{O(k^2 \log k)}$. For $x > 1 + 1/k$, $(3kx)^{k}$  is still less than $x^{O(k^2 \log k)}$.
\end{proofof}

%% file: main.bbl
\newcommand{\etalchar}[1]{$^{#1}$}
\begin{thebibliography}{AKM{\etalchar{+}}19}

\bibitem[AGS03]{AGS}
A.~Akavia, S.~Goldwasser, and S.~Safra.
\newblock Proving hard-core predicates using list decoding.
\newblock {\em FOCS}, 44:146--159, 2003.

\bibitem[AKM{\etalchar{+}}19]{AKMMVZ18}
Haim Avron, Michael Kapralov, Cameron Musco, Christopher Musco, Ameya
  Velingker, and Amir Zandieh.
\newblock A universal sampling method for reconstructing signals with simple
  fourier transforms.
\newblock In {\em Proceedings of the 51st annual ACM symposium on Theory of
  computing (STOC 2019)}, 2019.

\bibitem[BM86]{BM86}
Y.~Bresler and A.~Macovski.
\newblock Exact maximum likelihood parameter estimation of superimposed
  exponential signals in noise.
\newblock {\em IEEE Transactions on Acoustics, Speech, and Signal Processing},
  34(5):1081--1089, Oct 1986.

\bibitem[Che52]{chernoff1952}
Herman Chernoff.
\newblock A measure of asymptotic efficiency for tests of a hypothesis based on
  the sum of observations.
\newblock {\em The Annals of Mathematical Statistics}, 23:493--507, 1952.

\bibitem[CKPS16]{CKPS17}
Xue Chen, Daniel~M. Kane, Eric Price, and Zhao Song.
\newblock Fourier-sparse interpolation without a frequency gap.
\newblock In {\em Foundations of Computer Science(FOCS), 2016 IEEE 57th Annual
  Symposium on}, 2016.

\bibitem[CP18]{CP18}
Xue Chen and Eric Price.
\newblock Active regression via linear-sample sparsification.
\newblock {\em arXiv preprint arXiv:1711.10051}, 2018.

\bibitem[GGI{\etalchar{+}}02]{GGIMS}
Anna~C Gilbert, Sudipto Guha, Piotr Indyk, S~Muthukrishnan, and Martin Strauss.
\newblock Near-optimal sparse {F}ourier representations via sampling.
\newblock In {\em Proceedings of the thirty-fourth annual ACM symposium on
  Theory of computing}, pages 152--161. ACM, 2002.

\bibitem[GMS05]{GMS}
Anna~C Gilbert, S~Muthukrishnan, and Martin Strauss.
\newblock Improved time bounds for near-optimal sparse {F}ourier
  representations.
\newblock In {\em Optics \& Photonics 2005}, pages 59141A--59141A.
  International Society for Optics and Photonics, 2005.

\bibitem[HIKP12]{HIKP12}
Haitham Hassanieh, Piotr Indyk, Dina Katabi, and Eric Price.
\newblock Simple and practical algorithm for sparse {F}ourier transform.
\newblock In {\em Proceedings of the twenty-third annual ACM-SIAM symposium on
  Discrete Algorithms}, pages 1183--1194. SIAM, 2012.

\bibitem[IK14]{IK}
Piotr Indyk and Michael Kapralov.
\newblock Sample-optimal {F}ourier sampling in any constant dimension.
\newblock In {\em Foundations of Computer Science (FOCS), 2014 IEEE 55th Annual
  Symposium on}, pages 514--523. IEEE, 2014.

\bibitem[Man92]{Man92}
Y.~Mansour.
\newblock Randomized interpolation and approximation of sparse polynomials.
\newblock {\em ICALP}, 1992.

\bibitem[Moi15]{Moitra}
Ankur Moitra.
\newblock The threshold for super-resolution via extremal functions.
\newblock In {\em STOC}, 2015.

\bibitem[Pro95]{Prony}
R~Prony.
\newblock Essai experimental et analytique.
\newblock {\em J. de l’Ecole Polytechnique}, 1795.

\bibitem[PS15]{PS15}
Eric Price and Zhao Song.
\newblock A robust sparse {F}ourier transform in the continuous setting.
\newblock In {\em Foundations of Computer Science (FOCS), 2015 IEEE 56th Annual
  Symposium on}, pages 583--600. IEEE, 2015.

\end{thebibliography}
